\documentclass[11pt,letterpaper]{article}

\usepackage{ramp}

\def\authornotes{1}

\ifnum\authornotes=1
    \newcommand{\tselil}[1]{\footnote{\color{ForestGreen}Tselil: {#1}}}
    \newcommand{\misha}[1]{\footnote{\color{Orange}Misha: {#1}}}

    \newcommand{\Tnote}[1]{{\color{ForestGreen}[Tselil: #1]}}
    \newcommand{\mnote}[1]{{\color{Orange}[Misha: #1]}}
	\newcommand{\todo}[1]{{\color{Red} (TODO: #1)}}
\else
    \newcommand{\tselil}[1]{}
    \newcommand{\misha}[1]{}

    \newcommand{\Tnote}[1]{}
    \newcommand{\mnote}[1]{}

	\newcommand{\todo}[1]{}
\fi

\newcommand{\Ber}{\mathsf{Ber}}

\usepackage{scalerel}

\DeclareMathOperator{\GOE}{GOE}

\newcommand{\vAMP}{v_{\mathrm{AMP}}}

\DeclareMathOperator{\op}{\mathsf{op}}
\DeclareMathOperator{\plim}{\operatornamewithlimits{p-lim}}

\DeclareMathOperator{\PL}{PL}

\newcommand{\ind}{\boldsymbol{1}}

\newcommand{\sign}{\mathsf{sign}}

\DeclareMathOperator{\ReLu}{ReLu}

\DeclareMathOperator{\deloc}{del}

\author{Misha Ivkov\thanks{Stanford University. \texttt{mishai@stanford.edu}. Supported by NSF Graduate Research Fellowship.} \and Tselil Schramm\thanks{Stanford University.  \texttt{tselil@stanford.edu}. Supported by NSF CAREER award \# 2143246.}}

\title{Easy, robust approximate message passing for planted spike models}
\date{\today}

\begin{document}
\maketitle

\begin{abstract}
We present a simple and efficient algorithm for robust approximate message passing (AMP) in the spiked matrix setting.
In particular, let $\eps$ be a sufficiently small constant, and suppose that $X \in \R^{n \times n}$ is a Gaussian matrix with a planted rank-$1$ spike, and $E \in \R^{n \times n}$ is an adversarially chosen matrix supported on an $\eps n \times \eps n$ principal minor.
Let $v_{\mathrm{AMP}}(X)$ be the output of an AMP iteration on the uncorrupted matrix $X$.
We give a procedure that, given access only to the corrupted matrix $Y = X + E$, computes a vector $v_{\mathrm{ALG}}(Y)$ which is $\tilde{O}(\sqrt{\eps})$-close to $v_{\mathrm{AMP}}(X)$, for any of a class of AMP iterations which includes sparse Principal Component Analysis (PCA), non-negative PCA, and $\Z_2$ synchronization.
Our algorithm consists of a spectral pre-processing step combined with a robust spectral initialization procedure; given these inputs, we prove that (perhaps surprisingly) AMP is robust out-of-the-box.

\end{abstract}
\thispagestyle{empty}

\clearpage

{ \hypersetup{hidelinks} \tableofcontents }
\thispagestyle{empty}
\clearpage
\setcounter{page}{1}

\section{Introduction}

Many algorithmic problems that arise in high-dimensional statistics can be simply modeled as follows: we observe a matrix $X \in \R^{n \times n}$, where $X$ is the sum of a low-rank ``signal'' and i.i.d. entrywise ``noise'' from a well-behaved distribution (such as zero-mean Gaussian); in order to recover the low-rank signal, we compute a maximizer of the quadratic form $u^\top X u$ over vectors $u$ in some constraint set $K$.
Typically $X$ is a stand-in for the covariance matrix of observed data, and optimizing over $u \in K$ reflects the fact that we are looking for relations between covariates with a given structure, e.g. we may ask that $u \ge 0$ if we wish to recover positive associations.

A popular algorithm for this task is Approximate Message Passing (AMP), a family of algorithmic methods which generalize matrix power iteration.
The basic AMP algorithm starts from some initialization $x^{(0)} \in \R^n$ and computes a series of iterates by setting $x^{(t+1)} \approx f(Xx^{(t)})$,\footnote{The $\approx$ relation hides an additive term, the ``Onsager correction,'' which depends on $x^{(t)}$, but for the sake of simplicity we ignore this in the present discussion.} 
where the ``denoiser'' $f$ is a function (of the algorithm designers' choosing) from $\R \to \R$ applied coordinate-wise.
The goal of the ``powering'' action, $X x^{(t)}$, is to increase the quadratic form, while the denoiser $f$ is chosen to bring $f(Xx^{(t)})$ close to the constraint set $K$.

Though worst-case quadratic optimization problems are hopelessly hard, AMP algorithms have enjoyed enormous success in high-dimensional statistics, i.e. when $X$ is of the form ``signal + noise''.
The field of high-dimensional statistics has studied this case and its cousins extensively over the past decade, from which a rich and principled theory has emerged.
AMP algorithms were initially introduced as a generalization of Belief Propagation through the lens of statistical physics~\cite{Bol14, DMM09, BM11} but are now used for a variety of tasks, such as compressed sensing~\cite{DMM09}, sparse Principal Component Analysis (PCA)~\cite{DM14}, compact group reconstruction~\cite{DAM17,PWBM18,LFW23}, non-negative PCA~\cite{MR15}, and linear regression~\cite{DMM09, BM11, KMSSZ12}.
The upshot is that in this ``signal + noise'' setting, it is often possible to design denoiser functions that either achieve the information-theoretically optimal reconstruction of the spike or, at the very least, the best known computational bound.
We refer the reader to~\cite{FVRS22, MV21} for further examples and generalizations.

Many AMP algorithms are not natively robust, in the sense that small perturbations to the input can dramatically affect the output.
This is easily seen to be true in the worst case (where it follows from the $\NP$-hardness of quadratic optimization), but AMP is fragile even on the ``average case'' inputs $X$ of the form ``signal + low rank noise'', for example if the noise matrix has non-zero mean~\cite{CZK14,RSFS19,BG19}, or under additive perturbation by a structured spike of large norm~\cite[Example 1.10]{IS24}.
Indeed, even the matrix power iteration algorithm, which is more-or-less an AMP iteration with identity denoisers, is not robust, as small perturbations to the input can dramatically affect the singular vectors.

Prior works (\cite{IS24, IS25}) have established that, for a certain class of adversarial corruptions, it is possible to efficiently (with practical algorithms) recover vectors close to the AMP iteration on the \emph{uncorrupted} $X$.
In~\cite{IS25}, we show that AMP can be made robust to a class of corruptions called \emph{principal minor corruptions} (definition below) via a combination of spectral pre-processing and per-iteration outlier removal.
But the results of~\cite{IS24, IS25} are limited to the ``null'' model, in which the prior distribution over $X$ has no planted signal, and is purely composed of entrywise i.i.d. noise.

Unfortunately, the analysis of~\cite{IS24,IS25} does not extend to settings with a nontrivial spike.
For problems inspired by high-dimensional statistics, allowing a signal component in $X$ is essential; otherwise we cannot understand the capabilities of the algorithm for extracting signal from noise.

The main contribution of this work is an efficient (and indeed \emph{practical}) algorithmic procedure that makes AMP algorithms with a broad class of denoiser functions robust to principal minor corruptions, when the input comes from the more canonical ``\emph{signal} + noise'' model.
The core of our algorithm is a new robust algorithm for eigenvector recovery, which may be of independent interest.

\subsection{Setup \& Definitions}

\subsubsection{Spiked matrix model}

\begin{definition}[Spiked Gaussian matrix model]
	The \emph{Gaussian Orthogonal Ensemble} $\GOE(n)$ is the distribution over symmetric matrices $G$ with diagonal entries $G_{i,i}\sim N(0,2/n)$ and off diagonal entries $(G_{i,j} = G_{j,i}) \sim N(0,1/n)$ independently.

	We call a matrix a \emph{Spiked Gaussian} matrix if it is of the form $X = \frac{\lambda}{n} vv^\top + G$ for some $v\in \R^n$ with $\|v\|_2^2 = (1 \pm o_n(1))n$ and $G\sim \GOE(n)$.
	In this setting, we call $v$ the \emph{spike}.
\end{definition}

The celebrated work of Baik, Ben Arous, and P\'{e}ch\'{e}~\cite{BBP05} establishes that the spiked matrix model undergoes a phase transition.
On the one hand, when $\lambda < 1$, the spike is information theoretically ``invisible'' to the top eigenvector $\varphi_1(X)$ of $X$: in other words, $\langle v, \varphi_1(X)\rangle \rightarrow 0$ and $\|X\|_{\op} = 2 + o_n(1)$.
On the other hand, when $\lambda > 1$, the spike induces an outlier eigenvalue, and the top eigenvector is nontrivially correlated with the spike.
In particular, with high probability we have that $\langle v, \varphi_1(X)\rangle^2 \rightarrow 1 - \lambda^{-2}$ and $\|X\|_{\op} = \lambda + \frac 1\lambda + \widetilde O(n^{-1/2})$.\footnote{The $\widetilde O$ hides logarithmic terms in $n$; the fluctuations of the top eigenvalue look approximately Gaussian.}
At the same time, the rest of the spectrum of $X$ has the same limiting law as the GOE (the \emph{semicircle law}) and remains between $-2-\widetilde O(n^{-2/3})$ and $2+\widetilde O(n^{-2/3})$.\footnote{The fluctuations of the bulk come from the \emph{Tracy Widom} distribution.}
{\bf From here on out, we will denote the typical value of the largest eigenvalue in the spiked matrix model by $\widetilde \lambda \triangleq \lambda + \frac 1\lambda$ to reduce notation.}

We will consider the case most common in the AMP literature (cf.~\cite[Equation (M1)]{FVRS22}) wherein the spike $v$ is sampled i.i.d. coordinate-wise from some prior distribution $\pi_V$ with $\E[V^2] = 1$.\footnote{We will be mainly interested in ``nice'' distributions $V$ where it is the case that $\|v\|_2^2 = (1 \pm o_n(1))n$ with high probability.}
Throughout we will use $V$ to refer to a random variable sampled from $\pi_V$.

Our results will allow for any $V$ with nice tail behavior.
In particular, we will deal with subgaussian $V$.
A mean-$0$ random variable $\bX$ is said to be \emph{$\sigma$-subgaussian} if for each integer $k \in \N$, $\E[|\bX|^k] \le \sigma^k k^{k/2}$.
For example, a mean-$0$ Gaussian with variance $\sigma^2$ is $\sigma$-subgaussian, and a uniformly random sign $\in \{\pm 1\}$ is $1$-subgaussian.
Note that rescaling a $\sigma$-subgaussian variable $\bX$ to $C \bX$ for constant $C$ rescales the subgaussian parameter to $C\sigma$.
When $\bX$ is not mean-$0$, we will call it $\sigma$-subgaussian if $\bX - \E[\bX]$ is $\sigma$-subgaussian.

\subsubsection{Approximate Message Passing and example algorithmic problems}

\begin{definition}[AMP algorithm]\label{def:amp}
An \emph{Approximate Message Passing algorithm} is specified by a sequence of denoiser functions $\calF = f_0,f_1,f_2,\ldots$, with $f_t : \R \to \R$ for each $t \in \N$.
It takes as input a symmetric matrix $X\in \R^{n\times n}$, an initialization $v^0(X)$, a number of iterations $T \in \N$, and produces a pair of sequences of iterates $\widetilde x^0 \triangleq v^0, \widetilde x^1, \ldots,\widetilde x^{T-1}$ and $x^1, x^2, \ldots, x^T$ defined by
\[
\widetilde x^t = Xx^t - \Delta_t(x^{t-1}), \qquad x^{t+1} = f_t(\widetilde x^t)
\]
where $f_{t}$ is applied coordinate-wise, and $\Delta_t$ is the \emph{Onsager correction term} for decreasing correlations between iterates and is fully determined by $\calF$ (see~\pref{def:amp-iteration}).
AMP algorithms often also come with a \emph{rounding} procedure which is applied to the final iterate $x^T$, in order to ensure it satisfies the optimization constraints.
\end{definition}

Note that the above definition does not capture a fully general AMP iteration due to the restriction that $f_t$ be applied coordinate-wise; however, the majority of the AMP literature does use coordinate-wise denoisers.
The initialization $v^0(X)$ is typically chosen to have non-negligible correlation with $v$: depending on the problem, the usual choices are $v^0(X) \propto \E[V]\vec 1$ or $v^0(X) = \lambda\sqrt n \varphi_1(X)$.

We have chosen a few representative examples below, which have been studied in the high-dimensional statistics literature and which represent a varied set of spike structures and different levels of theoretical understanding (non-negative PCA is fully characterized, sparse PCA is the subject of some open problems).
In each example, $X$ comes from the Spiked Gaussian Model, with a different prior over the spike.

\begin{example}[non-negative PCA (nnPCA)]\label{eg:nnpca}
In the non-negative principal component analysis (nnPCA) problem, the spike has a prior distribution $\pi_V$ which is non-negative almost surely, and one is asked to maximize $x^\top X x$ over non-negative unit vectors $x \ge 0$.
The AMP algorithm which starts from $v^0(X) = \vec 1$
and uniformly chooses the separable denoiser $f_t(x) = \max(x,0)$ achieves the optimal quadratic form and optimal correlation with the signal (\cite{MR15}).\footnote{Technically $x^{T}$ may not be a unit vector nor non-negative, but AMP algorithms such as this one usually include a final ``rounding'' step---in this case, the rounding is just applying $f(x) = \max(x,0)$ followed by projection to the unit ball.}
In this case, up to the Onsager correction, AMP coincides with projected gradient ascent with ``infinite'' step size.
\end{example}

\begin{example}[$\Z_2$ synchronization]\label{eg:sync}
In the $\Z_2$ synchronization problem, the spike has a prior distribution $\pi_V$ which is uniform over $\{\pm 1\}$, and one is asked to maximize $x^\top X x$ over vectors $x$ with $x_i^2 = 1$.
The AMP algorithm which starts from $v^0(X) = \lambda\sqrt{n}\varphi_1(X)$ and uniformly chooses the separable denoiser $f_t(x) = \tanh(\lambda x)$ achieves the Bayes optimal quadratic form and correlation (\cite{DAM17}).\footnote{The rounding procedure here is to return $\sign(x^T)$, applied entrywise.}
\end{example}

\begin{example}[Sparse PCA]\label{eg:sparse-pca}
	In the Sparse PCA problem, the spike has a prior distribution $\pi_V$ which satisfies $\pi_V(\{0\}) \ge 1-\delta$, and one is asked to maximize $x^\top X x$ over $\delta$-sparse unit vectors $x$ (in other words, $\|x\|_0 \le \delta n$).
We will study the canonical special case where $V = \delta^{-1/2} \Ber(\delta)$ and $\delta < 1$ is sufficiently large; in this setting, Bayes-optimal denoisers are known (\cite{DM14}).\footnote{We mention that even if $\pi_V$ is unknown, a reasonable AMP iteration for this setting starts from $v^0(X) = \lambda\sqrt{n}\varphi_1(X)$ and uses the denoisers $f_t(x) = \sign(x) \cdot (|x| - \tau_t)_+$, for appropriately chosen $\tau_t$ (\cite{MV21}). We will not address this setting, as establishing whether the iterates are expanding/controlled is difficult to do without an explicit choice for $\tau_t$, and the optimal choice for $\tau_t$ is difficult to characterize theoretically but can be computed analytically (cf. the discussion surrounding~\cite[Proposition 2.1]{MV21}).}
\end{example}

With these examples in mind, we will consider two restricted classes of AMP iterations which capture a wide variety of results.
The first is the class of ``expanding'' AMP iterations, which at a high level are iterations in which the norm of the iterates grows sufficiently quickly with each step. 
We remark that this property is mediated by the denoiser functions; for example if the denoisers satisfy $\|f\|_\infty \ll \lambda$, the AMP iteration will not be expanding, whereas we will later show that $f(x) = \max(x,0)$ does expand.

\begin{definition}[Expanding AMP Iteration]\label{def:expanding}
	We call an AMP iteration $(L,c_1,c_2)$-\emph{expanding} for parameters $L\ge 1$ and $c_1,c_2\le 1$ if the Lipschitz constant of $f_t$ is at most $L$ for all $0\le t < T$, $f_t(0) = 0$ for all $0\le t < T$\footnote{This is satisfied by the nnPCA iteration. It is also satisfied by any Bayes-optimal AMP iteration when $V$ is symmetric.}, and furthermore,
\[
\Pr_{X}\left[
\begin{aligned}
\text{ for all }0 \le t < T,\,\, \|x^{t+1}\|_2^2 &\ge c_1^2 \|\widetilde x^t\|_2^2\\
\text{ and, for all }1 \le t < T,\,\, \|x^{t+1}\|_2^2 &\ge (c_2\lambda)^2 \|x^t\|_2^2
 \end{aligned}
\right] \ge 1-o_n(1).
\]
\end{definition}

To capture AMP algorithms which are not expanding, we introduce the notion of a controlled AMP iteration.
Again at a high level, a controlled AMP iteration is one where the denoiser functions are very flat in a neighborhood of $\supp(V)$.
This is a desirable property when the support of $V$ is discrete, because the denoiser $f$ is trying to softly ``round'' a continuous input to $\supp(V)$, so it is reasonable to map values in an interval around each $v \in \supp(V)$ very close to the value $v$.

\begin{definition}[Controlled AMP Iteration]\label{def:controlled}
	We call an AMP iteration $(B, c_1, c_2)$-\emph{controlled} for parameters $B,c_1,c_2 \ge 1$ with respect to $\lambda$ and $V$ if, for every $0\le t < T$,
	\begin{itemize}
		\item Each denoiser is totally bounded: $\sup_x |f_t(x)|\le B$.
		\item Each denoiser's Lipschitz constant is bounded: $\sup_{x,y} \frac{|f_t(x) - f_t(y)|}{|x - y|} \le B^2\lambda$.
		\item Each denoiser's \emph{local Lipschitz constant} is inverse-polynomially bounded in $\lambda$ whenever $\lambda \ge 4$: for each $v\in \supp(V)$, define the interval $I_v = \mu_t v \pm \frac {\mu_t} 5$; here $\mu_t$ is the expected correlation of $x^t$ with $v$. Typically it is true that $\mu_t$ is a constant proportional to $\lambda$ and must be computed as part of the AMP analysis; we define it formally in \pref{def:state-evolution}. 
We then require that for all $v \in \supp(V)$,
		\[\sup_{x,y\in I_v} \frac{|f_t(x) - f_t(y)|}{|x - y|} \le \frac{c_1}{\lambda}.\]
		\item When $\lambda \ge 4$, the signal correlation $\mu_t$ is large for all iterations: $\mu_t \ge \frac{\lambda}{2}$.
\item The final iterate has nontrivial norm:
		\[\Pr_{X}\left[ c_2\|x^{T}\|_2^2 \ge n\right] \ge 1 - o_n(1).\]
	\end{itemize}

\end{definition}

To our knowledge, these classes of iterates have not been considered in the past. 
Indeed, the conditions are not exactly self-evident, but they turn out to both (1) suffice to ensure that AMP, in combination with a spectral pre-processing step, is robust out-of-the-box, and (2) capture many canonical AMP-amenable spiked matrix models.
As we will show in appendix~\ref{sec:appendix-state-evolution}, each of the above examples is either expanding or controlled:

\begin{proposition}[Classification of AMP iterations]\torestate{
\label{prop:amp-classification}
	The following are true.
	\begin{itemize}
		\item The nnPCA AMP iteration is $(1,1/\sqrt 2,\E[V]/\sqrt{\pi})$-expanding.
		\item The $\Z_2$ synchronization AMP iteration is $(1,4,5^T)$-controlled.
		\item The Bayes-optimal Sparse PCA AMP iteration when $V = \delta^{-1/2} \Ber(\delta)$ is $(\frac{1}{\sqrt{\delta}}, 28, \frac{1}{\delta})$-controlled.
	\end{itemize}
}
\end{proposition}

In~\pref{prop:bayes-controlled}, we will further show that so long as the support of $V$ is finite, bounded, and well-separated, the Bayes-optimal AMP iteration from spectral initialization is controlled.

\subsubsection{Corruption model}

We will allow adversarially chosen corruptions in the \emph{$\eps$-principal minor} corruption model.

\begin{definition}[$\eps$-principal minor corruption]
Given matrices $X,Y \in \R^{n\times n}$, we say $Y$ is an \emph{$\eps$-principal minor corruption} of $X$ if $E\triangleq Y-X$ is supported on an $\eps n \times \eps n$-principal minor. We will often refer to the set of rows in this principal minor as $S^\ast$.

This is the corruption model considered in our prior works on robust AMP~\cite{IS24,IS25} as well as in prior works on robust community recovery in the stochastic block model~\cite{BMR21,MRW24}.
We will take the (mild) assumption throughout that $\eps = \widetilde \omega(n^{-1/2})$.
We discuss this model further in section~\ref{sec:discussion} below.

\end{definition}

\subsection{Results}

Our main theorem is the following.
We assume that $\lambda \ge \frac 1{\sqrt 2}$ (this specific value is the threshold for nnPCA).

\begin{theorem}[Informal combination of Theorems~\ref{thm:amp-expanding} and~\ref{thm:amp-controlled}]
\label{thm:informal-main}
	Suppose that $\calA$ is a $T$-step AMP algorithm on a Spiked Gaussian matrix $X$ where the signal $V$ is $\sigma$-subgaussian, and let $\vAMP(X)$ be the output of $\calA$ on $X$.
	Further, suppose that $Y$ is an $\eps$-principal minor corruption of $X$.
	\begin{itemize}
	\item If $\calA$ is a $(L,c_1,c_2)$-expanding iteration, there exists an algorithm taking $Y$ as input and producing in time $O(\eps n^4)$ a vector $\hat{v}(Y)$ such that
	\[\| \hat{v}(Y) - \vAMP(X)\|_2^2 \le O\left(\frac{L^3}{c_1c_2^2}\right)^{2T}\cdot \sigma^2 \eps \log \tfrac 1\eps \|\vAMP(X)\|_2^2.\]
	\item If $\calA$ is a $(B,c_1,c_2)$-controlled iteration with respect to $\lambda$ and $V$, then there exists an algorithm taking $Y$ as input and producing in time $O(\eps n^4)$ a vector $\hat{v}(Y)$ such that
	\[\| \hat{v}(Y) - \vAMP(X)\|_2^2 \le O(B^5c_1)^{2T} \cdot c_2\sigma^2 \eps \log^{2T + 1}\left(\tfrac 1\eps\right)\|\vAMP(X)\|_2^2.\]
	\end{itemize}
\end{theorem}

To elaborate, this theorem shows that, given a principal minor corruption $Y$ of $X$, it is still possible to recover a vector close to the AMP output on the uncorrupted $X$.
We recall from~\cite{IS24} that just running AMP directly on $Y$ need not always return something close to $X$ (as an example, consider nnPCA with an added large-norm sparse nonnegative spike).

\begin{remark}
As an immediate corollary, we obtain robust results for nnPCA, $\Z_2$ synchronization, and Sparse PCA.
It may be possible to improve the dependence on any of the parameters.
\end{remark}

Perhaps surprisingly, our analysis proceeds by showing that any expanding or controlled AMP iteration is \emph{already} robust, so long as the starting iterate is close to the starting iterate in the uncorrupted case.
For AMP problems in which $\E[V] = 0$, such as $\Z_2$ synchronization, it is common to choose the initialization $v^0(X) = \varphi_1(X)$.
Note that the top eigenvector is not a priori robust to principal minor corruptions; indeed, adding a huge number to a diagonal entry will cause the top eigenvector to become quite correlated with the corresponding indicator vector.
The workhorse of our robustification algorithm is a simple spectral algorithm for robust Principal Component Analysis (PCA) in the principal minor corruption model, which may be of independent interest.

\begin{theorem}[Principal-minor robust PCA]\torestate{
\label{thm:fast-robust-pca}
	Suppose that $Y$ is an $\eps$-principal minor corruption of $X$ with $\sigma$-subgaussian and $0$-mean signal distribution $V$.
	Then we can, in time $O(\eps n^4)$, compute a unit vector $w(Y)$ which satisfies
	\[\langle w(Y), \varphi_1(X)\rangle^2 \ge 1 - C\sigma^2 \eps \log \frac 1\eps \cdot  \left(\frac{\widetilde \lambda}{\widetilde \lambda - 2}\right)^3.\]
}
\end{theorem}
The algorithm is quite practical, as it consists only of eigenvector computations, sorting, and zeroing out matrix entries.
In the main body of the paper, we will also present a simple semidefinite programming algorithm which enjoys slightly better error guarantees.

\subsection{Discussion}
\label{sec:discussion}

We give a fast, spectral algorithm for many common planted AMP iterations under principal minor corruptions, thus extending the work of~\cite{IS25} beyond the null ``pure-noise'' model to the more statistically interesting ``signal+noise'' setting.
Our running time guarantees are precisely the same.
Essentially, the ``expanding'' and ``controlled'' conditions for AMP iterates allow us to prove that in the planted case, AMP is almost robust \emph{as-is}: one must simply first pre-process the input matrix $Y$ to remove excessively large eigenvalues, and make sure to initialize at an iterate which is $\eps$-close to the initialization on the uncorrupted matrix.

For expanding iterations with initialization $x^0 \propto \vec 1$ (e.g. nnPCA) our \emph{analysis} closely mirrors the null-case analysis of~\cite{IS25}, even though the algorithm is a bit different.
For controlled iterations ($\Z_2$ synchronization, sparse PCA), the analysis departs from that of~\cite{IS25} and makes much more careful use of properties of the denoiser functions and of state evolution.
Furthermore, in these cases the initialization is a function of the input.
We introduce a robust algorithm for computing top eigenvectors to ensure that the initial iterate is $\eps$-close to the top eigenvector of the uncorrupted matrix.
Perhaps this robust PCA algorithm will find additional uses.

Here we have considered principal minor corruption, a corruption model that has appeared in prior literature, both for robust AMP in the null model~\cite{IS24,IS25} and for robust community recovery in the stochastic block model~\cite{BMR21,MRW24}, a planted ``signal + noise'' model.
However, there are many other possible corruption models, and hence many natural directions for future research.

In particular, it is not clear whether AMP can be made robust against stronger corruption models in the \emph{planted} setting.
\cite[Observation 1.11]{IS24} showed that, in the null model, it is possible to replace $o(n^2)$ entries in $X$ such that it is information theoretically impossible to recover a vector close to the AMP output, and hence the principal minor restriction (or something like it) is necessary.
However, their lower bound does not apply to the planted setting, so it is possible that a different algorithm can handle the more general case of strong contamination.

Another interesting case to consider is that of an \emph{oblivious} adversary which can corrupt a constant fraction of entries, located anywhere in $X$ (not restricted to a principal minor).
The power of an oblivious adversary is not even understood in the null model, but the planted model may be easier to study because of the presence of the signal.

\subsection{Technical Overview}

The proof of theorem~\ref{thm:informal-main} is, at its core, a proof by induction.
We begin by producing $\hat{Y}$, a spectrally pre-processed version of the corrupted matrix $Y$, using the spectral cleaning procedure of~\cite{IS25}.
Given this, we show that as long as the AMP iteration is either expanding (definition~\ref{def:expanding}) or controlled (definition~\ref{def:controlled}), as long as we have access to an iterate $y^t$ which is close in $\ell_2$ distance to the iterate $x^t$ on the uncorrupted matrix, then computing the AMP iteration on $y^t$ with $\hat{Y}$ produces an iterate $y^{t+1}$ which is close in $\ell_2$ to $x^{t+1}$.
Each condition (expanding vs. controlled) requires a different inductive proof.

This is sufficient to finish the argument as long as we can compute an initialization $y^0$ which is close to $x^0$ in $\ell_2$.
This is immediate when $x^0$ does not depend on $X$, as is the case when $x^0 \propto \E[V]\vec 1$.
However, for many problems $\E[V] = 0$, and in these cases it is customary to choose $x^0 = \lambda \sqrt{n}\varphi_1(X)$.
Hence, we will need to develop a robust algorithm for computing the top eigenvector: given $Y$, we will be able to find a vector $y^0$ which is close to the top eigenvector of the \emph{uncorrupted} matrix $X$.
We call this task ``robust principal component analysis'' or ``robust PCA.''

In this overview, we will first discuss a simplified algorithm for principal-minor robust PCA and then show some of the high-level details behind the inductive argument for expanding and controlled iterations.

\subsubsection{Principal Minor Robust PCA}

In this section, we will discuss the proof of the following result:

\restatetheorem{thm:fast-robust-pca}

At a high level, we use the fact that the top eigenvector of $X$ is \emph{delocalized}.
In particular, we say a vector $x$ is $(\beta,\eps)$-delocalized if every restriction $x_S$ of $x$ to a set $S\subset[n]$ with $|S| \le \eps n$ satisfies $\|x_S\|_2^2\le \beta$.
For example, the vector $\frac 1{\sqrt n} \vec 1$ (which is entirely flat) is $\eps$-delocalized, and the same holds true for a $\Z_2$ spike.
In fact, any vector of independent $\sigma$-subgaussian random variables is $O(\sigma^2\eps \log \frac 1\eps)$-delocalized with high probability.

With this in mind, the algorithm and analysis are as follows:

\begin{itemize}
	\item While the top eigenvector $w$ of $Y$ is not sufficiently delocalized, we \emph{clean} $Y$.
That is, let $S$ be the collection of the largest-magnitude $\eps n$ coordinates (which must be localized by assumption).
We sample $i\in S$ with probability proportional to $w_i^2$ and delete the $i$th row and column from $Y$.
Then we continue.

		To analyze this step, suppose that the adversarial perturbation $E$ is supported on some collection of coordinates $S^\ast$.
		Using the spectral properties of random matrices, we prove that, outside of the coordinates in $S^\ast$, $w$ is with very high probability $(\delta,\eps)$-delocalized (for $\delta \asymp \sigma^2 \eps \log \frac 1\eps \left(\frac{\widetilde \lambda}{\widetilde \lambda - 2}\right)^2$).
		Then, if $w$ is not $2\delta$-delocalized, it follows that
		\[\frac{\sum_{i\in S\cap S^\ast} w_i^2}{\sum_{i\in S} w_i^2} = 1 - \frac{\sum_{i\in S - S^\ast}w_i^2}{\sum_{i\in S} w_i^2}\ge 1 - \frac{\delta}{\sum_{i\in S} w_i^2} \ge \frac 12\]
		since $|S - S^\ast| \le \eps n$.
		This implies that, with probability at least $\frac 12$ at each iteration, we remove some $i\in S^\ast$.
		Thus, this while loop terminates after at most $4\eps n$ iterations with high probability (if $E$ has been entirely removed, then $w$ is $\delta$-delocalized).

	\item After completing the while loop, we have a $2\delta$-delocalized vector $w$ which has a large quadratic form, at least $\widetilde \lambda (1-O(\eps))$, with a $(1-O(\eps))$-measure principal minor of $Y$.
We show that because it is delocalized, this vector cannot pick up too much in a quadratic form with any surviving portion of $E$ and additionally satisfies $w^\top X w \ge \widetilde \lambda (1 - O(\delta))$.
	Since $X$ has an eigenvalue gap, in the sense that $\lambda_2(X) \le 2 + o_n(1)$, we may use this large quadratic form to show that $\langle w, \varphi_1(X)\rangle^2 \ge 1 - O(\delta) \cdot \frac{\widetilde \lambda}{\widetilde \lambda - 2}$.
This completes the proof.

\end{itemize}

\subsubsection{Closeness of expanding and controlled AMP iterations}

For $(L,c_1,c_2)$-expanding AMP iterations, we can essentially adapt the inductive argument of~\cite{IS25}; in fact, because of the expanding condition, we can simplify the algorithm of~\cite{IS25} and avoid the ``clipping'' step for the iterates.
To elaborate, let us suppose for simplicity that $\|Y\|_{\op} \le 5\widetilde \lambda$ (modulo some additional details, we can run the spectral cleaning algorithm of~\cite{IS25} in order to ensure that is the case) which then implies that $\|E\|_{\op} = \|X-Y\|_{\op} \le 6 \widetilde{\lambda}$.
Then, let $y^t$ denote the $t^{\text{th}}$ iterate resulting from running AMP on $Y$.
We aim to show inductively that $\|y^{t+1} - x^{t+1}\|_2^2 \le \left(\frac{C\cdot L}{c_1c_2}\right)^{2(t+1)} \sigma^2 \eps \log \frac{1}{\eps} \|x^{t+1}\|_2^2$, for $C$ a large universal constant.
Ignoring the Onsager correction, we may begin to unroll the expressions for $x^{t+1},y^{t+1}$:
\begin{align*}
\|y^{t+1} - x^{t+1}\|_2^2 &= \|f_t(Yy^t) - f_t(Xx^t)\|_2^2\\
&\le L^2 \|Yy^t - Xx^t\|_2^2 & \text{(applying Lipschitzness of $f_t$)}\\
&\le 2L^2 \left(\|Y(y^t - x^t)\|_2^2 + \|(Y - X)x^t\|_2^2 \right) &\text{(using the Almost-Triangle Inequality)}
\end{align*}
Let's discuss these two terms separately.
Since $\|Y\|_{\op} \le 5\widetilde \lambda$, we can apply the inductive hypothesis on the first term:
\[\|Y(y^t - x^t)\|_2^2 \le \|Y\|_{\op}^2 \cdot \|y^t - x^t\|_2^2 \le (5\widetilde \lambda)^2\cdot \left(\frac{C L}{c_1c_2}\right)^{2t} \sigma^2 \eps \log \frac 1\eps \|x^t\|_2^2.\]
For the second term, recall that $Y - X = E$, and $E$ is supported on a principal minor consisting of at most $\eps n$ rows and columns; we will let $S$ be that collection of indices.
Then, it follows that
\[\|(Y - X)x^t\|_2^2 = \|Ex^t\|_2^2 \le \|E\|_{\op}^2 \cdot \|x^t_S\|_2^2 \le (6\widetilde \lambda)^2 \cdot \|x^t_S\|_2^2.\]
To handle this latter term, we make use of \emph{state evolution}, which is a characterization of AMP iterates at the core of much of the theory of AMP algorithms.
In particular, state evolution implies that the AMP iterates on the uncorrupted matrix coarsely behave like the unknown spike with additive Gaussian noise.
In particular, for any univariate function $h:\R\to\R$, denote by $\bar{h}:\R^n \to \R$ the average of $h$ on its input, $\bar{h}(u) = \frac{1}{n}\sum_{i=1}^n h(u_i)$.
State evolution guarantees that for any sufficiently well-behaved $h$, with high probability $\bar{h}(X x^t) \approx \bar{h}(\mu_t \cdot v + \sigma_t g)$, where $v$ is the planted spike, $g \sim \cN(0,\Id)$, and $\mu_t,\sigma_t$ are scalar parameters (which depend on the choice of AMP iteration) satisfying $\mu_t = \Theta(\lambda)$ and $\sigma_t = \Theta(1)$.

When $f_t$ is sufficiently nice (as is the case for a $(L,c_1,c_2)$-expanding iteration), we can prove that $x^t = f_t(X x^t)$ is \emph{delocalized} in a similar sense to the previous section: no subset of at most $\eps n$ coordinates can capture a large fraction of $x^t$'s mass.
More precisely, by setting $h = f_t \cdot \Ind_{I_t}$ for a well-chosen interval $I_t \subset \R$, we can show that, for all candidate sets $R$ with $|R| \le \eps n$,
\[\|x^t_R\|_2^2 \le \frac{32}{c_1^2}\sigma^2 \eps \log \frac 1\eps \cdot \|x^t\|_2^2\]
with probability $1 - o_n(1)$ over the original choice of $X$ (and in particular this holds for $R = S$).
Putting everything together, it follows that
\begin{align*}
\|y^{t+1} - x^{t+1}\|_2^2 &\le 2L^2 \left(25\widetilde \lambda^2 \cdot \left(\frac{C L}{c_1c_2}\right)^{2t} \sigma^2 \eps \log \frac 1\eps\|x^t\|_2^2 + 36\widetilde \lambda^2\cdot \frac{32}{c_1^2}\sigma^2 \eps \log \frac 1\eps \cdot \|x^t\|_2^2 \right)	\\
&= \left(2L^2 \left[25 \left(\frac{C L}{c_1c_2}\right)^{2t} + \frac{36\cdot 32}{c_1^2}\right]\right) \left(\sigma^2 \eps \log \frac 1\eps\right)\cdot \biggl(\widetilde \lambda^2 \cdot \|x^t\|_2^2\biggr)\\
&\le \frac{CL^2}{c_1^2}\cdot \left(\frac{C L}{c_1c_2}\right)^{2t}\left(\sigma^2 \eps \log \frac 1\eps\right)\cdot \biggl(\widetilde \lambda^2 \cdot \|x^t\|_2^2\biggr), \\
\intertext{as long as $C$ is a sufficiently large constant. Now we would be done, except that we have an undesirable factor of $\widetilde{\lambda}^2$ on the right hand side. But because our iteration is expanding, by definition, $\widetilde\lambda\|x^t\|$ is comparable to $\|x^{t+1}\|$, and so we can finish the proof:}
&\le \frac{C L^2}{c_1^2}\cdot \left(\frac{CL}{c_1c_2}\right)^{2t}\left(\sigma^2 \eps \log \frac 1\eps\right)\cdot \left(\frac 1{c_2^2}\|x^{t+1}\|_2^2\right) \\
&= \left(\frac{C L^2}{c_1c_2}\right)^{2(t+1)}\sigma^2 \eps \log \frac 1\eps \|x^{t+1}\|_2^2.
\end{align*}
This is exactly what we set out to prove, so the inductive step is complete.
We can extend this analysis, similarly to~\cite{IS25}, to work with the Onsager correction (and to accommodate the fact that we must work with a matrix $\hat Y$ that is the output of a spectral cleaning procedure on $Y$).\\

Unfortunately, this particular analysis can fail when an iteration is not expanding, because 1) we may have $\|x^{t+1}\| \le \|x^t\|$, and suffer factors polynomial in $\lambda$ and 2) $x^t$ may not be sufficiently delocalized.
However, we can still say something nontrivial for the class of $(B,c_1,c_2)$-\emph{controlled} iterations; in effect, these iterations are ones which have a very small Lipschitz constant in the neighborhood around the support of $V$, and we will have to use state evolution to show that the majority of each iterate's coordinates are within this neighborhood.

For controlled iterations, our goal will be to prove the guarantee
$$
\|y^{t+1} - x^{t+1}\|_2^2 \le (C B^4 c_1)^{2(t+1)} c_2\sigma^2 \eps \log^{2(t+1)+1} \frac 1\eps \|x^{t+1}\|_2^2.
$$
By nontriviality of controlled iterates, it will be sufficient to prove the statement
\[\frac 1n \|y^{t+1} - x^{t+1}\|_2^2 \le (C B^4 c_1)^{2(t+1)} \sigma^2 \eps \log^{2(t+1)+1} \frac 1\eps\]
(note the removal of $c_2$), by directly applying $c_2\|x^{t+1}\|_2^2 \ge n$.

As before the proof is by induction; the first step is to remove the effect of $E$ on the current iteration (incurring an $O(\eps)$ loss in accuracy).
Recalling that $S$ is the collection of rows in $E$'s support, we begin by splitting into these ``bad'' coordinates:
\begin{align*}
	\|y^{t+1} - x^{t+1}\|_2^2 &= \|\left(y^{t+1} - x^{t+1}\right)_S\|_2^2 + \|\left(y^{t+1} - x^{t+1}\right)_{\overline S}\|_2^2\\
	&= \sum_{i\in S} (f_t((X x^t)_i) - f_t((Yy^t)_i))^2 + \|\left(y^{t+1} - x^{t+1}\right)_{\overline S}\|_2^2\\
	&\le |S| \cdot (2B)^2 + \|\left(y^{t+1} - x^{t+1}\right)_{\overline S}\|_2^2,\\
	&= 4B^2 \eps n + \|\left(y^{t+1} - x^{t+1}\right)_{\overline S}\|_2^2,
\end{align*}
where we have used that $f_t$ is bounded.
Now it remains to control the coordinates outside of the support of $E$.
For the coordinates $i \in \bar{S}$ we have $(Yy^t)_i = (Xy^t)_i$ because $X,Y$ agree on each such row $i$.
Ignoring the Onsager correction once more and recalling that $\|X\|_{\op} = \widetilde \lambda \le 4\lambda$ when $\lambda \ge \frac 1{\sqrt 2}$, by induction
\[\|(Yy^t - Xx^t)_{\overline{S}}\|_2^2
= \|(X(y^t - x^t))_{\overline{S}}\|_2^2
 \le \|X\|_{\op}^2 \cdot \|y^t - x^t\|_2^2
\le 16\lambda^2 \delta_t^2 n\]
where we have used the shorthand $\delta_t^2 = (C B^2 c_1)^{2t} \sigma^2 \eps \log^{2t+1} \frac 1\eps$.

Now we consider a case analysis.
If $\lambda \le 10^2 B\sqrt{\log \eps^{-1}}$, we can immediately complete the proof using the Lipschitzness of $f_t$:
\[
\|y^{t+1}_{\overline S} - x^{t+1}_{\overline S}\|_2^2
= \|\left(f_t(Y y^t) - f_t(Xx^t)\right)_{\overline{S}}\|_2^2
	\le (\lambda B^2)^2 \| \left(Yy^t- X x^t\right)_{\overline{S}}\|_2^2\le \lambda^4  B^4 \delta_t^2 \le 4\cdot 10^8 B^8 \log^2 \frac 1\eps \cdot \delta_t^2 n,\]
thus verifying the induction hypothesis for $t+1$ when $C$ is a sufficiently large constant.

Otherwise, $\lambda \gg B\sqrt{\log\eps^{-1}}$.
In this scenario, we invoke state evolution with the function $h$ an indicator of the neighborhood around $\mu_t v$ to show that all but at most $2\eps n$ coordinates of $Xx^t$ satisfy the closeness condition $|(Xx^t)_i - \mu_t v_i| \lesssim \sigma_t\sqrt{\log \frac 1\eps} \ll \mu_t$, with probability $1 - o_n(1)$ over $X$.
We can then argue for coordinates $i$ satisfying the closeness condition, $(Y y^t)_i$ and $(Xx^t)_i$ lie within the \emph{small Lipschitz constant} regime guaranteed by being a controlled iteration.
Indeed, if $R$ is the set of coordinates in $\overline{S}$ that do not satisfy the closeness condition, since $R$ contains at most $2\eps n$ coordinates, we can incur the $B$-Lipschitz constant cost on it in the same way that we did with $S$:
\[\|\left(y^{t+1} - x^{t+1}\right)_{\overline{S}}\|_2^2 \le 8B^2 \eps n + \|\left(y^{t+1} - x^{t+1}\right)_{\overline R}\|_2^2.\]
Now, there may be coordinates $i \in \overline{R}$ such that $(Xx^t)_i$ is in the small Lipschitz constant regime whereas $(Yy^t)_i$ is not; however, since we have shown $\|\left(Yy^t - Xx^t\right)_{\overline{S}}\|_2^2 \le 4\lambda^2 \delta_t^2 n$, by an averaging argument the number of coordinates $i$ where $((Yy^t)_i - (Xx^t)_i)^2 > \left(\frac{\lambda}{20}\right)^2$ is at most $20^2\cdot 4\delta_t^2 n$, and hence the error contribution of these coordinates is also bounded.

In summary, we can restrict our attention to coordinates $i$ such that $| (Xx^t)_i - \mu_t v_i| \ll \mu_t$, and by the triangle inequality, the same holds for $y$: $|(Yy^t)_i - \mu_t v_i| \ll \mu_t$.
Thus, both terms lie within the same interval $I_v$ and enjoy the inverse-polynomial local Lipschitz constant.
We may now apply this small Lipschitz constant:
\begin{align*}
\|y^{t+1} - x^{t+1}\|_2^2
&= \sum_{i=1}^{n} (f_t((Y y^t)_i) - f_t((Xx^t)_i))^2 \\
&\le \sum_{i=1}^{n} \left(\frac{c_1}{\lambda}\right)^2 ((Y y^t)_i - (Xx^t)_i)^2 \\
&\le \left(\frac{c_1}{\lambda}\right)^2\cdot \|\left(Yy^t - Xx^t\right)_{\overline{S}}\|_2^2  \\
&\le 16c_1^2 \delta_t^2 n
\end{align*}
which completes the argument.

\section{Preliminaries}

\subsection{Approximate Message Passing}

\begin{definition}[AMP Iteration]
\label{def:amp-iteration}
Let $\{f_t\}_{t\ge 0}$ be a sequence of Lipschitz functions on $\R$ with weak derivatives $\{f'_t\}_{t \ge 0}$.
Given an input matrix $X$ and with the initial values $x^{-1} = \widetilde x^{-1} = \vec 0$ and $\widetilde x^0(X)$ with $\|\widetilde x^0(X)\|_2^2 = \lambda^2 n$, we define the AMP iteration via 
\[x^{t+1} = f_t(\widetilde x^t),\qquad b_t = \frac 1n \sum_{i=1}^{n} f'_t(\widetilde x^t), \qquad \widetilde x^t = Xx^t - b_{t-1} x^{t-1}.\]
The term $b_{t-1} x^{t-1}$ is known as the \emph{Onsager correction}, and ensures that the iterates remain uncorrelated.
\end{definition}

AMP iterates obey \emph{State Evolution}, which precisely characterizes the limiting distribution of the average of any function over the iterates' coordinates.

\begin{definition}[State Evolution (cf.~\cite{FVRS22}, Theorem 3.1)]
\label{def:state-evolution}
	Suppose that $X = \frac \lambda n vv^\top + G$, where $G\sim \GOE(n)$ and $v_i \sim V$ for some univariate distribution with $\E[V^2] = 1$.
	For an AMP iteration $\{f_t\}$, we define the two mutually recursive quantities
	\[\mu_{t+1} = \lambda \E[V f_t(\mu_t V + \sigma_t Z)], \qquad \sigma_{t+1}^2 = \E[f_t(\mu_t V + \sigma_t Z)^2].\]
	If, for every $\PL(2)$ function $\psi: \R^2 \rightarrow \R$, the initial iterate obeys the constraint 
	\[\plim_{n\rightarrow\infty}\frac 1n \sum_{i=1}^{n} \psi(\widetilde x^0_i; v_i) \rightarrow \E[\psi(\mu_0 V + \sigma_0 Z; V)],\]
	then the same holds for all future iterations:
	\[\plim_{n\rightarrow \infty} \frac 1n \sum_{i=1}^{n} \psi(\widetilde x^t_i; v_i) \rightarrow \E[\psi(\mu_t V + \sigma_t Z; V)]\]
\end{definition}

It is well known that, by State Evolution, whenever $f_t'$ can be approximated by $\PL(2)$ functions, we can substitute $B_{t-1} = \E[f_t'(\mu_t V + \sigma_t Z)]$ for $b_{t-1}$ in the AMP iteration.
We also define the \emph{Bayes-AMP} iteration.

\begin{definition}[Bayes-AMP (cf.~\cite{FVRS22}, Corollary 3.9 and Theorem 3.10)]
	The \emph{Bayes-AMP} iteration for a given signal distribution $\pi_V$ is the iteration which chooses denoisers
	\[f_t(y) = \E[V \mid \mu_t V + \sigma_t Z = y].\]
\end{definition}

It turns out that Bayes-AMP iterations are controlled whenever $\supp(V)$ is finite and well-separated; we prove this proposition in appendix~\ref{sec:appendix-state-evolution}.
\begin{proposition}[Bayes-AMP for finite and well-separated $V$ is controlled]
\torestate{
\label{prop:bayes-controlled}
	Suppose that $\supp(V) = \{v_1, v_2, \ldots, v_m\}$ is finite, bounded by some $B$, and satisfies $\Delta \triangleq \min_{i\neq j} |v_i - v_j| \ge 1$. Then, the Bayes-AMP iteration on $V$ is $(B, \frac{28B^2}{\Delta^4}\cdot \max \frac{1}{p_i}, 1/\E[V]^2)$-controlled, if started at a spectral initialization.
}
\end{proposition}

\begin{remark}
	When $\E[V] = 0$, such as in $\Z_2$-synchronization, this result gives a trivial bound on $c_2$. In this scenario, we can instead derive a more problem-specific bound which depends directly on $\mu_T$ and $\sigma_T$.
\end{remark}

\subsection{Delocalized functions}

A vector is said to be \emph{delocalized} if it is relatively ``flat,'' in the sense that no small subset of coordinates contains too much of the $\ell_2$ mass. 
There are several distinct notions of delocalization; here, we make use of the following definition.

\begin{definition}[Delocalized Vectors]
\label{def:deloc-vec}
	We say a vector $w$ is $(\beta,\eps)$-delocalized if
	\[\max_{|S|\le\eps n} \sum_{i\in S} w_i^2 \le \beta\cdot \|w\|_2^2.\]
	
	We will define the \emph{delocalization bound} of a vector $w$ via $\deloc_\eps(w) = \max\limits_{|S|\le\eps n} \sum_{i\in S} w_i^2$.
\end{definition}

Note that $\deloc_\eps$ satisfies the Almost-Triangle Inequality.
It turns out that as long as the spike's distribution $\pi_V$ is subgaussian, all AMP iterates are delocalized.

\begin{proposition}[AMP iterates are delocalized]
\torestate{
\label{prop:deloc-amp}
	Suppose that $V$ is $\sigma$-subgaussian. 
	Then the AMP iterate $\widetilde x^t$ is $(32\sigma^2 \eps \log \frac 1\eps,\eps)$-delocalized.
	Furthermore, if the AMP iteration is $(L,c_1,c_2)$-expanding, then $x^t$ is also $(\frac{32L^2}{c_1^2}\sigma^2 \eps \log \frac 1\eps,\eps)$-delocalized.
	}
\end{proposition}

We prove this fact in appendix~\ref{sec:amp-deloc}.

\subsection{Semidefinite Programming}

Let $A\bullet B = \tr(A^\top B) = \sum_{i,j} A_{i,j} B_{i,j}$.
A semidefinite program (SDP) over a variable $X \in \R^{n\times n}$ and with constraint matrices $A_1, A_2, \ldots, A_k$ is a program of the form
\begin{align*}
	\max \qquad & C\bullet X\\
	s.t. \qquad & X\succeq 0\\
	& A_i \bullet X \le b_i \qquad \forall 1\le i\le k
\end{align*}
We will assume that $\{A_i\}_{1\le i\le k}$ and $C$ are symmetric matrices.

As is standard, by binary search over objective values $c = C \bullet X$ one may remove the maximization objective and replace it with a linear constraint of the form $- C \bullet X \le - c$.
The problem of optimizing then reduces to the problem of finding a feasible solution to the set of linear constraints.

\begin{definition}[Separation Oracle]
	A separation oracle for a particular SDP is a polynomial time procedure $\calO$ which satisfies the following two properties for an input $Y\in \R^{n\times n}$: 
	\begin{itemize}[noitemsep]
		\item If $Y$ is feasible for the SDP, $\calO(Y) = \bot$.
		\item Else, $\calO(Y)$ returns some hyperplane $(A, b)$ such that $A\bullet Y > b$ but every feasible solution $X$ of the SDP satisfies $A\bullet X \le b$.
	\end{itemize}
\end{definition}

Given a separation oracle, the SDP can be approximately optimized in polynomial time (cf.~\cite[Theorem 3.10]{FKP19})\footnote{There are some technical details which we skip over: the runtime of the optimization depends on $\log(R/r\delta)$, where $\delta$ is accuracy to which the optimization is computed, $R$ is an upper bound on the radius of the feasible domain for $X$, and $r$ is the largest size of a sphere contained within our feasible region. 
In our setting, these will be polynomially (resp. inverse polynomially) bounded.}. 
We will make use of SDPs with a possibly exponential number $k$ of linear constraints, but which nonetheless admit a polynomial-time separation oracle and are thus polynomial time solvable.
For more on SDPs, we refer the reader to~\cite[Section 3]{FKP19}.

\section{Making Planted AMP Robust to Principal Minor Corruptions}

In this section, we will prove two main theorems, considering the robustness of planted AMP under either expanding or controlled iterations.
In both of the below theorems, we will write $X = \frac \lambda n vv^\top + G$ and $Y = X + E$ where $E$ is an $\eps$-principal minor corruption of $X$, and $v_i \sim \pi_V$ i.i.d for some $\sigma$-subgaussian distribution $\pi_V$ with $\E[V^2] = 1$.
Furthermore, recall that $\|X\|_{\op} = \lambda + \frac 1{\lambda} \pm o_n(1)$.
We will write $\widetilde \lambda = \lambda + \frac 1\lambda$ for this top eigenvalue, and let $C$ be a constant.

\begin{theorem}[Robust Planted AMP under Expanding Iterations]\torestate{
\label{thm:amp-expanding}
	Take the definitions of $X$ and $Y$ from above. 
	Suppose that the AMP iteration $f^0, f^1, \ldots,f^T$ is $(L,c_1,c_2)$-expanding and that there is some initialization $\widetilde y^0$ such that $\|\widetilde y^0 - \widetilde x^0\|_2^2 \le \eta \|\widetilde x^0\|_2^2$. 
	Then, there is an algorithm $\calA$ running in time $O(\eps n^4)$, taking as input $Y$, and returning a vector $\calA(Y)$ satisfying
	\[\|\calA(Y) - \vAMP(X)\|_2^2 \le O\left(\frac{L^3}{c_1c_2^2}\right)^{2T} \max\left(\sigma^2 \eps \log \frac 1\eps,\eta\right) \cdot \|\vAMP(X)\|_2^2.\]
}
\end{theorem}

\begin{theorem}[Robust Planted AMP under Controlled Iterations]\torestate{
\label{thm:amp-controlled}
	Take the definitions of $X$ and $Y$ from above. 
	Suppose that the AMP iteration $f^0, f^1, \ldots,f^T$ is $(B, c_1,c_2)$-controlled and that there is some initialization $\widetilde y^0$ such that $\|\widetilde y^0 - \widetilde x^0\|_2^2 \le \eta \|\widetilde x^0\|_2^2$.
	Then,
	\[\|\vAMP(Y) - \vAMP(X)\|_2^2 \le O(B^5c_1)^{2T} c_2\max(\eps,\eta)  \log^{2T} \frac 1\eps \cdot \|\vAMP(X)\|_2^2.\]
}
\end{theorem}

We will prove both of these results in the upcoming sections. 
Before doing so, we require one additional primitive to ensure access to some $\widetilde{y}^0$ which is close to $\widetilde{x}^0$ in the case that $\E[V] = 0$.
When $\E[V] = 0$, AMP is typically initialized from an input-dependent warm start before running the iteration: in particular, it starts at $x^0 = \lambda \sqrt{n}\varphi_1(X)$, for $\varphi_1$ the unit eigenvector of the largest eigenvalue.
The next theorem proves that we can efficiently compute an approximation to $\varphi_1(X)$ when only given $Y$.

\begin{theorem}[Robust PCA]\torestate{
\label{thm:robust-pca}
	Given a $\eps$-principal minor corruption $Y$ of $X$ with $\sigma$-subgaussian $0$-mean signal distribution $V$, we can in polynomial time compute a unit vector $w(Y)$ which satisfies 
	\[\langle w(Y), \varphi_1(X)\rangle^2 \ge 1 - C\sigma^2 \eps \log \frac 1\eps  \frac{\widetilde \lambda}{\widetilde \lambda - 2}.\]
}
\end{theorem}

The running time of the above algorithm is not practically feasible ($\widetilde O(n^8)$, due to having to solve an SDP with the ellipsoid method), so we also give a fast spectral version of this algorithm with slightly worse error guarantees.

\restatetheorem{thm:fast-robust-pca}

	Note that the spectral algorithm's error has a worse dependence on $\widetilde{\lambda}$ than the SDP-based algorithm; it is possible that this is an artifact of our analysis.

\subsection{Principal-minor robust PCA}

Our goal in this section is to prove~\pref{thm:robust-pca} and~\pref{thm:fast-robust-pca}. {\bf For simplicity throughout this section, let $\alpha = 32\sigma^2 \eps \log \frac 1\eps$.}

We will make use of a spectral algorithm from~\cite{IS25} which returns a large principal minor of bounded operator norm.

\begin{lemma}[Spectral cleaning of principal minor corruptions (cf. Lemma 3.5 of~\cite{IS25})]
\label{lem:spectral-cleaning}
	Given $Y = X + E$ as above, there is an algorithm running in time $\widetilde O(n^3)$ and with high probability returning a symmetric matrix $\hat{Y}$ and collection of indices $T\subseteq [n]$ such that
	\begin{itemize}[noitemsep]
		\item $|T| \ge (1 - 4\eps)n$ and $\hat{Y}_{T,T} = Y_{T,T}$. Furthermore, $\hat{Y}$ is zero outside of this principal minor.
		\item The norm of $\hat{Y}$ is controlled: $\|\hat{Y}\|_{\op} \le 5\|X\|_{\op}$.
	\end{itemize}
\end{lemma}

\begin{algorithmSELF}[SDP reconstruction of top eigenvector]
\label{alg:robust-pca}
	\textbf{Input: } Matrix $Y$ guaranteed to be of the form $Y = X + E$ for a principal minor corruption $E$.
	
	\noindent \textbf{Operation: }
	\begin{itemize}[noitemsep]
		\item Using the algorithm referenced in~\pref{lem:spectral-cleaning}, compute the spectral cleaning $(\hat{Y}, T)$ of $Y$, satisfying $|T| \ge (1-4\eps n)$ and $\|\hat Y \|_{\op} \le 5\|X\|_{\op}$.
		\item Solve the SDP
		\begin{align*}
		M^* = \argmax_M \qquad & \hat{Y} \bullet M	\\
		\text{s.t} \qquad & M \succeq 0, \quad \tr(M) = 1, \quad \max_{|S|\le\eps n} \sum_{i\in S} M_{i,i} \le 2\alpha.
		\end{align*}

		\item Compute $w$, the top unit eigenvector of $M^*$.
	\end{itemize}
	
	\noindent \textbf{Output: } The unit vector $w$, which satisfies $\langle w, \varphi_1(X) \rangle^2 \ge 1 - 30\alpha \cdot \frac{\widetilde \lambda}{\widetilde \lambda - 2}$ (\pref{thm:fast-robust-pca}).
\end{algorithmSELF}

We show in Appendix~\ref{sec:amp-deloc} that $\varphi_1(X)$ is delocalized and achieves high quadratic form with $\hat{Y}$ (which will allow us to show feasibility of the above system). 

\begin{proposition}[$X$ correlates well with $\hat{Y}$]\torestate{
\label{prop:deloc-eig}
	Suppose that $\hat{Y}$ is a principal minor of $Y = X + E$ with at least $(1 - c\eps)n$ rows and satisfying $\|\hat{Y}\|_{\op} \le 5\widetilde \lambda$. Then, $\varphi_1(X)^\top \hat{Y} \varphi_1(X) \ge \widetilde \lambda \left(1 - (6 + 3c)\alpha\right)$.
}
\end{proposition}

\begin{proof}[Proof of Theorem~\ref{thm:robust-pca}]
We will run~\pref{alg:robust-pca} on $Y$ and output the result.
Note that if $T$ is the set of nonzero rows/columns of $\hat Y$, the trace constraint ensures that $M^*$ is supported on $T$.

\textbf{Runtime. }
We begin by noting that the above SDP has a polynomial time separation oracle, even though there are exponentially many constraints. 
Indeed, to check the last condition, it suffices to sort the diagonal entries of $M$ in $O(n\log n)$ time, and if the sum of the largest $\eps n$ diagonal coordinates exceeds $2\alpha$, return the hyperplane just consisting of the $1$s on those diagonal coordinates.
Therefore, we can solve the above SDP in polynomial time with the ellipsoid method.

\textbf{SDP feasibility. } 
Let $x = \varphi_1(X)$.
Recall that $|\overline T| \le 4\eps n$ with high probability, and thus if a vector $y$ is $(\alpha,\eps)$-delocalized, it must be the case that $\|y_T\|_2^2 \ge 1 - 4\alpha$.
Take  $M = x_T x_T^{\top}$, appropriately renormalized to have $\tr(M) = 1$.
Then, by~\pref{prop:deloc-eig}, it follows that $\max_{|S|=\eps n} \sum_{i\in S} M_{i,i} \le \frac{\alpha}{1 - 4\alpha} \le 2\alpha$.
This implies that the system is feasible, and has objective value at least $x^\top \hat{Y} x \ge \widetilde \lambda (1 - 18\alpha)$.

 \textbf{Output guarantee. }
Letting $\lambda_1(M)$ be the top eigenvalue of $M$, it follows that $1\ge \lambda_1(M) \ge x^\top M x$.
	Therefore,
	\[\langle x,\varphi_1(M)\rangle^2 = \frac 1{\lambda_1}\left(x^\top M x - \sum_{i > 1} \lambda_i(M) \langle x, \varphi_i(M)\rangle^2\right) \ge x^\top M x - (1 - \lambda_1(M)) \ge 2x^\top M x - 1.\]
	This implies that we just need to show that $x^\top Mx$ is large (equivalently, this shows that $M$ is approximately rank $1$, with top eigenvector near $x$).
	
	Let $M'$ be the restriction of $M$ to the rows and columns in $S^\ast \cap T$, and note that $M'$ is still PSD. 
	Furthermore, by the SDP constraints we know that $\tr(M') \le 2\alpha$ since $|S^\ast| \le \eps n$.
	This implies that for any matrix $A$, $M'\bullet A \le \|A\|_{\op} \cdot 2\alpha$ (by taking the spectral decomposition of $M'$ and using the trace constraint).
	We can thus apply this to $A = E_{T,T}$: note that $\|E_{T,T}\|_{\op} \le \|\hat Y\|_{\op} + \|X\|_{\op}\le 6\widetilde \lambda + o_n(1)$.
	Therefore,
	\[M\bullet X = M\bullet \hat Y - M'\bullet E_{T,T} \ge \widetilde \lambda (1 - 18\alpha) - \|E_{T,T}\|_{\op} \cdot \alpha \ge \widetilde \lambda (1 - 30\alpha) - o_n(1).\]
	Next, define $\beta_i = \varphi_i(X)^\top M \varphi_i(X)$ and note that each $\beta_i$ is nonnegative. We have that
	\begin{align*}
		\widetilde \lambda (1 - 30\alpha) - o_n(1) &\le M\bullet X\\
		& = \sum \beta_i \lambda_i(X) \\
		&= \beta_1 \widetilde \lambda + \sum_{i > 1} \beta_i \lambda_i (X) \\
		&\le \beta_1 (\widetilde \lambda - o_n(1)) - (1 - \beta_1)\cdot (2 + o_n(1))\\
		& = \beta_1(\widetilde \lambda - 2) + 2 + o_n(1)
	\end{align*}
	Here, we used the fact that, other than the spike eigenvalue, all of $X$'s eigenvalues are at most $2 + o_n(1)$ with high probability, and $\|X\|_{\op} \le \widetilde \lambda + o_n(1)$. Rearranging, it follows that
	\[\beta_1 = x^\top M x \ge \frac{\widetilde \lambda (1 - 30\alpha) - 2 - o_n(1)}{\widetilde \lambda - 2 - o_n(1)} \ge 1 - 31\alpha\frac{\widetilde \lambda}{\widetilde \lambda - 2}\]
	where we used that $\eps = \widetilde \omega(n^{-1/2})$ and that the fluctuations of the top eigenvector of $X$ are $\widetilde O(n^{-1/2})$ with high probability.
	\end{proof}
	
The above algorithm requires us to solve an SDP with the ellipsoid method (which is polynomial time), but the running time is still prohibitively slow.
We next give a spectral algorithm whose analysis takes advantage of similar phenomena, but incurs a slightly larger error in the small-$\lambda$ regime.
In the below analysis, we will not worry about the $o_n(1)$ terms: similar discussion as in the SDP case show that they do not affect the final result.

\restatetheorem{thm:fast-robust-pca}

\begin{algorithmSELF}[Spectral reconstruction of top eigenvector]
\label{alg:fast-robust-pca}
\textbf{Input: } Matrix $Y$ guaranteed to be of the form $Y = X + E$ for a principal minor corruption $E$.

\noindent \textbf{Operation: }
\begin{itemize}[noitemsep]
	\item Using~\pref{lem:spectral-cleaning}, compute the spectral cleaning $\hat Y$ of $Y$.
	\item While $w\triangleq \varphi_1(\hat Y)$ is not $\left(2C\alpha \left(\frac{\widetilde \lambda}{\widetilde \lambda - 2}\right)^2, \eps\right)$-delocalized:
	\begin{itemize}[noitemsep]
		\item Let $S$ be the top $\eps n$ coordinates of $w$. Sample $i\in S$ with probability proportional to $w_i^2$, and remove both row and column $i$ from $\hat{Y}$.
	\end{itemize}
	\item Return $w$.
\end{itemize}

\noindent \textbf{Output: } Unit vector $w$ such that $\langle w,\varphi_1(X)\rangle^2 \ge 1 - O\left(\alpha \cdot \left(\frac{\widetilde\lambda}{\widetilde \lambda - 2}\right)^3\right)$.
\end{algorithmSELF}

We will require the following result (recalling that $\alpha = 32\sigma^2 \eps \log \frac 1\eps$), which will be proven in \pref{sec:amp-deloc}.

\begin{proposition}[Delocalization of off-corruption coordinates]
\torestate{
\label{prop:deloc-off-corruption}
	Suppose that $X$ is a spiked Gaussian matrix with $\sigma$-subgaussian signal distribution $V$, and let $A$ be a principal minor of $X$ with at least $(1 - c\eps)n$ rows. Further, let $B = A + E$, where $E$ is an $\eps$-principal minor corruption of $A$ supported on the rows $S$ and having $\|E\|_{\op} \le 6\widetilde \lambda$, and let $w$ be the top eigenvector of $B$. Then,
	\[\max_{|U|\le \eps n, U\cap S = \varnothing} \sum_{i\in U} w_i^2 \le Cc\alpha \left(\frac{\widetilde \lambda}{\widetilde \lambda - 2}\right)^2.\]
}
\end{proposition}

\begin{proof}[Proof of Theorem~\ref{thm:fast-robust-pca}]
We will run~\pref{alg:fast-robust-pca} on $Y$ and output the result.

\textbf{Runtime. } We prove below that the number of iterations of the while loop is at most $O(\eps n)$ with high probability, and computing the top eigenvector takes time $O(n^3)$.

\textbf{Analysis. } First, note that when all rows of $E$ have been removed and at most $8\eps n$ total rows have been removed,~\pref{prop:deloc-off-corruption} implies the desired delocalization (since $S = \varnothing$), and therefore the while loop will terminate.
Therefore, a valid stopping condition is that all rows of $E$ get removed.
We prove that, otherwise,
\begin{equation}\frac{\sum_{i\in S\cap S^\ast} w_i^2}{\sum_{i\in S} w_i^2} \ge \frac 12. \label{eq:progress}\end{equation}
As in the analysis of the spectral cleaning algorithm from \cite{IS24}, this implies that as long as delocalization does not hold we have probability at least $\frac{1}{2}$ of removing an element of $S^\ast$, thus by a Martingale argument the number of extra rows we will remove is at most $4\eps n$ with high probability, implying that a total of at most $8\eps n$ rows are removed.

In this scenario, by~\pref{prop:deloc-eig}, it follows that $w^\top \hat{Y} w \ge \widetilde \lambda(1 - 30\alpha)$. Using delocalization, we thus know that
\[w^\top X w = w^\top \hat{Y} w - w^\top E w \ge \widetilde \lambda(1 - 30\alpha) - 6\widetilde \lambda \cdot 2C\alpha \left(\frac{\widetilde \lambda}{\widetilde \lambda - 2}\right)^2 = \widetilde \lambda \left(1 - C\alpha \left(\frac{\widetilde \lambda}{\widetilde \lambda - 2}\right)^2\right).\]
Expanding $w$ in the $\varphi_i(X)$ basis as $w = \sum \beta_i \varphi_i(X)$, we see that
\[w^\top X w = \sum \beta_i^2 \lambda_i(X) \le \beta_1^2 \widetilde \lambda + (1 - \beta_1^2)\cdot 2 = 2 + \beta_1^2(\widetilde \lambda - 2)\]
which upon rearranging yields
\[\beta_1^2 \ge \frac{\widetilde \lambda \left(1 - C\alpha \left(\frac{\widetilde \lambda}{\widetilde \lambda - 2}\right)^2\right) - 2}{\widetilde \lambda - 2} = 1 - C\alpha \left(\frac{\widetilde \lambda}{\widetilde \lambda - 2}\right)^3\]
as desired.

Thus, it remains to prove~\pref{eq:progress}. 
Towards this end, we will denote $\delta = 8C\alpha \left(\frac{\widetilde \lambda}{\widetilde \lambda - 2}\right)^2$: then, we know that $\deloc_{\eps}(w) > 2\delta$ while $\deloc_{\eps}(w_{\overline {S^\ast}}) \le \delta$ by~\pref{prop:deloc-off-corruption}.
This immediately implies that
\[\frac{\sum_{i\in S\cap S^\ast} w_i^2}{\sum_{i\in S} w_i^2} = \frac{\sum_{i\in S} w_i^2 - \sum_{i\in S - S^\ast} w_i^2}{\sum_{i\in S} w_i^2} \ge \frac{\sum_{i\in S} w_i^2 - \delta}{\sum_{i\in S} w_i^2} = 1- \frac{\delta}{\sum_{i\in S} w_i^2} \ge \frac 12\]
as desired. 
\end{proof}

\subsection{Robust AMP for expanding iterations}

In this section, we will prove the following theorem.

\restatetheorem{thm:amp-expanding}

\begin{proof}
 The argument will mimic that of~\cite[Lemma 3.6]{IS25}, with no clipping (recall that clipping is only required for polynomial denoisers), and we will proceed by induction.
 For convenience, we will introduce $\alpha = \max(\sigma^2 \eps \log \frac 1\eps,\eta)$ and assume the inductive hypothesis that
 \[\|y^{t} - x^{t}\|_2^2 \le \left(\frac{10^5 L^3}{c_1 c_2^2}\right)^{2t} \alpha \|x^t\|_2^2.\]
	 By~\pref{lem:spectral-cleaning}, we begin by constructing the $4\eps$-spectral cleaning $\hat{Y}$ of $X$ with $\|\hat Y\|_{\op} \le 5\|X\|_{\op}$.
	 We will further assume without loss of generality that errors are monotonically increasing: that is, $ \|y^{t-1} - x^{t-1}\|_2^2 \le \|y^t - x^t\|_2^2$.
	 It will also be useful to recall that $|B_{t-1}| \le L$, by definition of $B_t$.
	 
	 We may now, by the Almost-Triangle Inequality and the Lipschitzness of $f_t$, write
	 \begin{align*}
	 	\|y^{t+1} - x^{t+1}\|_2^2 &= \left\|f_t(\hat Y y^t - B_{t-1} y^{t-1}) - f_t(Xx^t - B_{t-1} x^{t-1})\right\|_2^2\\
	 	&\le L^2 \|(\hat Y y^t - B_{t-1} y^{t-1}) - (Xx^t - B_{t-1}x^{t-1})\|_2^2\\
	 	&\le 2L^2 \left(\|\hat Yy^t - Xx^t\|_2^2 + |B_{t-1}|^2 \|y^{t-1} - x^{t-1}\|_2^2\right)\\
	 	&\le 2L^2 \|\hat Yy^t - Xx^t\|_2^2 + 2L^4 \|y^{t} - x^{t}\|_2^2 \stepcounter{equation}\tag{\theequation}\label{eq:induction}
	 \end{align*}
	 It remains to handle the first term.
	 To do so, we may expand
	 \[\|\hat Yy^t - Xx^t\|_2^2 = \|\hat Yy^t - \hat Yx^t + \hat Yx^t - Xx^t\|_2^2 \le \|\hat{Y}\|_{\op}^2\|y^t - x^t\|_2^2 + \|(\hat Y - X)x^t\|_2^2.\]
	 wherein we can substitute $\|\hat{Y}\|_{\op} \le 5\widetilde \lambda$.
	 More interesting is the second term, where we have to use the structure of $\hat{Y} - X$, which looks like (up to permutation of rows)
	 \[\hat{Y} - X = \begin{bmatrix}
 	 & F_1 &  \\
\hline
 	F_2 \vline & \hat E & {\bf 0}\\
 	\phantom{F_2} \vline & {\bf 0} & {\bf 0}	
 \end{bmatrix}
\]
where $\hat E$ is the portion of the perturbation $E$ not removed by spectral cleaning, and $F_1$ and $F_2$ represent the non-corrupted entries removed by spectral cleaning and equal $-X$ on their support, and have row or column support at most $4\eps n$, respectively.
By~\pref{prop:deloc-amp}, we know that $x^t$ is $(\frac{32}{c_1^2}\alpha,\eps)$-delocalized and thus also $(\frac{128}{c_1^2}\alpha, 4\eps)$-delocalized. 
Therefore, since $\|\hat E\|_{\op} \le \|\hat Y\|_{\op} + \|X\|_{\op} \le 6\widetilde \lambda$ and $\|F_2\|_{\op} \le \|X\|_{\op} \le \widetilde \lambda$, it follows that
\begin{align*} 
\|(\hat E + F_2)x^t\|_2^2 &\le \|\hat E + F_2\|_{\op}^2 \cdot \frac{128L^2}{c_1^2}\alpha \|x^t\|_2^2 \le (7\widetilde \lambda)^2 \cdot \frac{128L^2}{c_1^2}\alpha \|x^t\|_2^2.
\end{align*}
Finally, we see that, if $R$ is the row-support of $F_1$, then $F_1 x^t = -(Xx^t)_{R} = -(\widetilde x^{t} + B_{t-1} x^{t-1})_{R}$.
By the Almost-Triangle Inequality, we see that
\[\|F_1 x^t\|_2^2 \le 2\left(\|\widetilde x^t_R\|_2^2 + \|B_{t-1} x^{t-1}_R\|_2^2\right) \le 256 \alpha \|\widetilde x^t\|_2^2 + L^2 \cdot \frac{256L^2}{c_1^2}\alpha \|x^{t-1}\|_2^2. \]

Putting everything together into~\pref{eq:induction} and recalling that $\|x^{t-1}\|_2^2 \le \frac 2{c_2^2} \|x^t\|_2^2$ whenever $\lambda \ge \frac 1{\sqrt 2}$, we see that
\begin{align*}
 \|y^{t+1} - x^{t+1}\|_2^2 &\le 2L^2 \|\hat Yy^t - Xx^t\|_2^2 + 2L^4 \|y^{t} - x^{t}\|_2^2\\
 &\le 2L^2 \left((5\widetilde \lambda)^2 \|y^t - x^t\|_2^2 + (7\widetilde \lambda)^2 \cdot \frac{256L^2}{c_1^2}\alpha \|x^t\|_2^2 + 512\alpha \left(\|\widetilde x^t\|_2^2 + \frac{L^4}{c_1^2}\|x^{t-1}\|_2^2\right)\right) + 2L^4 \|y^t - x^t\|_2^2\\
 &\le (2L^2 (5\widetilde \lambda)^2 + 2L^4) \|y^t - x^t\|_2^2 + \left((7\widetilde \lambda)^2\cdot \frac{256L^4}{c_1^2} + \frac{1024L^6}{c_1^2 c_2^2}\right) \alpha \|x^t\|_2^2 + 512L^2\alpha \|\widetilde x^t\|_2^2\\
 &\le 100 L^4 \widetilde \lambda^2 \|y^t - x^t\|_2^2 + 2^{18} \frac{L^6}{c_1^2c_2^2} \widetilde \lambda^2 \alpha \|x^t\|_2^2 + 512L^2 \alpha \|\widetilde x^t\|_2^2.\\
 \intertext{Now we apply induction on $\|y^t - x^t\|_2^2$:}
 &\le 100L^4 \widetilde \lambda^2 \cdot \left(\frac{10^{5} L^3}{c_1c_2^2}\right)^{2t} \alpha \|x^t\|_2^2 +  2^{18} \frac{L^6}{c_1^2c_2^2} \widetilde \lambda^2 \alpha \|x^t\|_2^2 + 512L^2 \alpha \|\widetilde x^t\|_2^2\\
 &\le \left(\frac{10^4 L^3}{c_1c_2}\right)^2 \left(\frac{10^{5} L^3}{c_1c_2^2}\right)^{2t}\cdot \alpha \widetilde \lambda^2 \|x^t\|_2^2 + 512L^2 \alpha \|\widetilde x^t\|_2^2\\
 \intertext{Finally, we use that $\|\widetilde x^t\|_2^2 \le \frac 1{c_1^2} \|x^{t+1}\|_2^2$ and $\widetilde \lambda^2 \|x^t\|_2^2 \le \frac{1}{16c_2^2} \|x^{t+1}\|_2^2$:}
 &\le \left(\left(\frac{10^4 L^3}{c_1c_2}\right)^2\cdot \frac{1}{c_2^2} + \frac{512L^2}{c_1^2}\right) \left(\frac{10^{5} L^3}{c_1c_2^2}\right)^{2t} \alpha \|x^{t+1}\|_2^2\\
 &\le \left(\frac{10^{5} L^3}{c_1c_2^2}\right)^{2(t+1)}\alpha \|x^{t+1}\|_2^2
\end{align*}
as desired.
\end{proof}

\subsection{Robust AMP for controlled iterations}

In this section, we will prove the following theorem.

\restatetheorem{thm:amp-controlled}

At a high level, it turns out that AMP on controlled functions is natively robust to principal minor corruptions, outside of the initialization.

One way to reason about this intuitively is to restrict to the case when $V$ is Rademacher, and $f_t(x) = \tanh(x)$. 
Then, consider doing one iteration of AMP on $X$ versus $Y$, starting from the same initialization. 
Even if $E$ has large norm, the fact that $\tanh$ is bounded implies that after applying it, the total error is bounded by $\eps$.

\begin{proof}[Proof of Theorem~\ref{thm:amp-controlled}]
	We may drop the $c_2 \cdot \frac 1n \|x^t\|_2^2$ term by the definition of a controlled iteration; it suffices to show that
	\begin{equation}\frac 1n\|y^{t} - x^{t}\|_2^2 \le (10^5 B^5 c_1)^{2t} \max(\eps,\eta) \log^{2t} \frac 1\eps \triangleq \delta_t^2  \label{eq:controlled-induction}\end{equation}
	for all $1\le t\le T$. 
	We will require the following claim, which we prove at the end of this section. 
	\begin{claim}[Controlled denoisers are not too expanding]
	\label{clm:controlled-contracting}
		Suppose that $\widetilde y$ is any vector such that $\frac 1n \|\widetilde y - \widetilde x^t\|_2^2 \le (\beta \lambda)^2$. 
		Then, $\frac 1n \| f_t(\widetilde y) - x^{t+1}\|_2^2 \le 10^8(B^8c_1^2) \beta^2 \log^2 \frac 1\eps$. 
	\end{claim}
	
	Since $\frac 1n\|\widetilde x^0\|_2^2 = \lambda^2$ by assumption, this claim implies that $\frac 1n \|y^1 - x^1\|_2^2 \le 10^8 (B^8 c_1^2) \max(\eps,\eta) \log^2 \frac 1\eps$.
	The rest of the proof is by induction on $t$, so suppose that we have proven~\pref{eq:controlled-induction} at this particular $t$.
	
	Let us assume that errors only increase: that is $\|y^t - x^t\|_2^2 \ge \|y^{t-1} - x^{t-1}\|_2^2$, and recall that $S\triangleq S^\ast$ is the collection of coordinates on which $E$ is supported.
	
	Then, we can write
	\begin{align*}
	  \frac 1n \|y^{t+1} - x^{t+1}\|_2^2 &= \frac 1n \|y^{t+1}_S - x^{t+1}_S\|_2^2 + \frac 1n \|y^{t+1}_{\overline S} - x^{t+1}_{\overline S}\|_2^2 \\
	  &= \frac 1n \sum_{i\in S} (f_t(\widetilde y^t_i) - f_t(\widetilde x^t_i))^2 + \frac 1n \|y^{t+1}_{\overline S} - x^{t+1}_{\overline S}\|_2^2\\
	  &\le 4B^2 \eps + \frac 1n \|y^{t+1}_{\overline S} - x^{t+1}_{\overline S}\|_2^2 \stepcounter{equation}\tag{\theequation}\label{eq:split-s}
	\end{align*}
	since $f_t$ is bounded by $B$ and $|S|\le \eps n$. 
	To bound this latter term, let us look directly at $\widetilde y$ and $\widetilde x$.
	Since $Y_{\overline S, [n]} = X_{\overline S, [n]}$, by definition of the AMP iterates it follows that
	\begin{align*}
		\frac 1n \|\widetilde y^t_{\overline S} - \widetilde x^t_{\overline S}\|_2^2 &= \frac 1n \left\|(Xy^t - B_{t-1}y^{t-1})_{\overline S} - (Xx^t - B_{t-1} x^{t-1})_{\overline S}\right\|_2^2\\
		&\le \frac 1n \left\|(Xy^t - B_{t-1}y^{t-1}) - (Xx^t - B_{t-1} x^{t-1})\right\|_2^2\\
		&\le 2\left(\frac 1n \|X(y^t - x^t)\|_2^2 + \frac 1n \|B_{t-1}(y^{t-1} - x^{t-1})\|_2^2\right)\\
		&\le 2(\|X\|_{\op}^2 + B_{t-1}^2) \cdot \frac 1n \|y^t - x^t\|_2^2\\
		&\le 2(\widetilde \lambda^2 + B^4\lambda^2) \cdot \delta_t^2\\
		&\le (2B^2 \lambda \delta_t)^2.
	\end{align*}
	We next define $\widetilde y$ to equal $\widetilde y^t$ on $\overline S$ and to equal $\widetilde x^t$ on $S$.
	Then, by Claim~\ref{clm:controlled-contracting} we see that
	\[\frac 1n\|y^{t+1}_{\overline S} - x^{t+1}_{\overline S}\|_2^2 = \frac 1n\|f_t(\widetilde y) - f_t(\widetilde x^t)\|_2^2 \le 10^8 (B^8 c_1^2) (2B^2 \delta_t)^2 \log^2 \frac 1\eps = 4\cdot 10^5 (B^{10} c_1^2) \delta_t^2\log^2 \frac 1\eps.\]
	Plugging this in to~\pref{eq:split-s}, it follows that
	\begin{align*}
	\frac 1n \|y^{t+1} - x^{t+1}\|_2^2 &\le 4B^2 \eps + \frac 1n \|y^{t+1}_{\overline S} - x^{t+1}_{\overline S}\|_2^2\\
	& \le 4B^2\eps + 4\cdot 10^8 (B^{10} c_1^2) \delta_t^2\log^2 \frac 1\eps\\
	&\le 8\cdot 10^8 (B^{10} c_1^2) \delta_t^2 \log^2 \frac 1\eps\\
	&\le \delta_{t+1}^2	
	\end{align*}
	which completes the induction.
\end{proof}

We finish by proving the claim. 
	\begin{proof}[Proof of Claim~\ref{clm:controlled-contracting}]
	We case on whether $\lambda > 40B\sqrt {\log \frac 1\eps}$ or not; recall that $\sigma_t \le B$.
	If $\lambda \le 40B\sqrt {\log \frac 1\eps}$, then we have immediately proved the claim:
	\[\frac 1n \|f_t(\widetilde y) - f_t(\widetilde x^{t})\|_2^2 \le (\lambda B^2)^2 \cdot \frac 1n \|\widetilde y - \widetilde x^t\|_2^2 \le \lambda^4 B^4 \beta^2 \le 40^4B^8 \beta^2  \log^2 \tfrac 1\eps. \]
	
	The brunt of this proof is in the case when $\lambda$ is large. 
	Indeed, let $F_i$ be the event that $|\widetilde x^t_i - \mu_t v_i| > 2\sigma_t\sqrt{\log \frac 1\eps}$. 
	Then, we know by State Evolution that
	\[\frac 1n \sum_{i=1}^{n} \ind[F_i] \rightarrow \Pr\left[(\mu_t V +\sigma_t Z - \mu_t V)^2 > 4\sigma_t^2\log \tfrac 1\eps\right] \le \Pr\left[|Z| \ge 2\sqrt{\log \tfrac 1\eps}\right] \le 2\eps.\]
	Further, let $S^\dagger = \{i\in [n] : (\widetilde y_i - \widetilde x^t_i)^2 > (\lambda/20)^2\}$, and note that $\frac 1n |S^\dagger| \le (20\beta)^2$ by averaging.
		Putting these together, it follows that
	\begin{align*}
		\frac 1n \|f_t(\widetilde y) - x^{t+1}\|_2^2 &= \frac 1n \sum_{i=1}^{n} (f_t(\widetilde y_i) - f_t(\widetilde x^t_i))^2 \ind[F_i \lor i\in S^\dagger] + \frac 1n \sum_{i=1}^{n} (f_t(\widetilde y_i) - f_t(\widetilde x^t_i))^2 \ind[\overline{F_i} \land i\notin S^\dagger]\\
		&\le (2B)^2 \cdot \frac 1n(|S^\dagger| + \sum_{i=1}^{n} \ind[F_i]) + \frac 1n \sum_{i=1}^{n} (f_t(\widetilde y_i) - f_t(\widetilde x^t_i))^2 \ind[\overline{F_i} \land i\notin S^\dagger]\\
		&\le 4B^2 ((20\beta)^2 + 2\eps) + \frac 1n \sum_{i=1}^{n} (f_t(\widetilde y_i) - f_t(\widetilde x^t_i))^2 \ind[\overline{F_i} \land i\notin S^\dagger]\\
		&\le 50^2 B^2 \beta^2 + \frac 1n \sum_{i=1}^{n} (f_t(\widetilde y_i) - f_t(\widetilde x^t_i))^2 \ind[\overline{F_i} \land i\notin S^\dagger]
	\end{align*}
	where we used that $\beta^2 \ge \eps$.
	
	It remains to handle this latter sum. 
	Here, for each $i$, note that by the triangle inequality we know that $|y_i - \mu_t v_i| \le 2\sigma_t \sqrt{\log \frac 1\eps} + \frac \lambda{20} \le \frac \lambda{10} \le \frac {\mu_t}{5}$
	since $\mu_t \ge \frac \lambda 2$ when $\lambda \ge 10B$.

	Therefore, since $\widetilde x^t_i$ satisfies the same guarantees and recalling that the Local Lipschitz constant is at most $\frac{c_1}{\lambda}$, it follows that
	\[\frac 1n \sum_{i=1}^{n} (f_t(\widetilde y_i) - f_t(\widetilde x^t_i))^2 \ind[\overline{F_i} \land i\notin S^\dagger] \le \left(\frac{c_1}{\lambda}\right)^2 \frac 1n \|\widetilde y - \widetilde x^t\|_2^2 \le \left(\frac{c_1}{\lambda}\right)^2 \cdot (\beta\lambda)^2 \le c_1^2 \beta^2.\]
	Putting these two facts together, we have shown that
	\[\frac 1n \|f_t(\widetilde y) - x^{t+1}\|_2^2 \le 50B^2 \beta^2 + c_1^2 \beta^2 \le 40^4 (B^8 c_1^2) \beta^2 \log^2 \frac 1\eps\]
	as desired.\footnote{When $\lambda$ is this large, this analysis implies that we can derive an inductive bound without the extra $\log \frac 1\eps$ factors.}
	\end{proof}

\bibliographystyle{alpha}
\bibliography{main}

 \appendix

 \addcontentsline{toc}{section}{Appendices}
 
 \section{Characteristics of AMP iterates and Spiked Matrices} 

\subsection{Classification of AMP Iterations}\label{sec:appendix-state-evolution}

In this subsection, we prove~\pref{prop:bayes-controlled} and~\pref{prop:amp-classification}.

\restateprop{prop:bayes-controlled}

\begin{remark}
	If we define $\Delta_i = \min_{j\neq i} |v_i - v_j|$ while keeping $\min\Delta_i \ge 1$, we can achieve a slightly stronger bound than the one obtained here: in particular, we will be $\left(B, 28B^2 \cdot \max\frac{1}{p_i \Delta_i^4}, 1/\E[V]^2\right)$-controlled. 
	This can be useful if outlier events are rare, while non-outlier events are fairly clustered together but of non-negligible probability.
\end{remark}

\begin{proof}[Proof of Proposition~\ref{prop:bayes-controlled}]
	We recall that in Bayes-AMP, denoisers are chosen such that $f_t(y) = \E[V \mid \mu_t V + \sigma_t Z = y]$. 
	As a consequence, note that $\mu_{t+1} = \lambda \E[V f_t(\mu_t V + \sigma_t Z)]= \lambda \E[f_t(\mu_t V + \sigma_t Z)^2] = \lambda \sigma_{t+1}^2$.
	
	Certainly, $|f_t(y)| \le B$ (it is never optimal to exceed the support of $V$), which gives the first condition.
	Furthermore, we may compute $f_t'(y) = \frac{\mu_t}{\sigma_t^2} \Var(V\mid \mu_t V + \sigma_t Z = y)$. 
	By Popoviciu's Inequality, this variance is bounded by $B^2$, so $|f_t'(y)|\le \lambda B^2$.
	
		Now, we show that $\mu_t \ge \frac \lambda 2$ when $\lambda \ge 4$; recall that $\mu_0 = \sqrt{\lambda^2 - 1} \ge \frac \lambda 2$ for a spectral initialization.
	We claim that this continues to hold inductively.
	Begin by noting that $\E[(\mu_t V + \sigma_t Z)f_t(\mu_t V + \sigma_t Z)] = \E[\mu_t V^2] = \mu_t$. 
	Thus, by Cauchy Schwarz,
	\[\mu_t^2 = \E[(\mu_t V + \sigma_t Z)f_t(\mu_t V + \sigma_t Z)]^2 \le \E[(\mu_t V + \sigma_t Z)^2] \sigma_{t+1}^2 = (\mu_t^2 + \sigma_t^2) \sigma_{t+1}^2\]
	Note that $\sigma_0^2 = 1$, and for $t\ge 1$, $\sigma_t^2 = \frac{\mu_t}{\lambda}$ by optimality of $f_t$.
	Then, we can write
	\[\sigma_{t+1}^2 \ge \frac{\mu_t^2}{\mu_t^2 + \sigma_t^2} \ge \frac{\mu_t}{\mu_t + 2} \ge \frac 12.\]
	We finish by using once more that $\mu_{t+1} = \lambda \sigma_{t+1}^2$.\\
	
	Next, we prove the local Lipschitz condition.
	For a given $v_i\in \supp(V)$, let us rewrite our interval as $I = \mu_t [v_i - \frac 15, v_i + \frac 15]$, and fix some $s\in [v_i - \frac 15, v_i + \frac 15]$.
	Suppose that we prove that
	\[\Pr[V \neq v_i \mid \mu_t V + \sigma_t Z = \mu_t s] \le \frac{1}{p_i} \exp\left(-\frac{3\mu_t \lambda \Delta^2}{10}\right).\]
	If this is true, then it follows that
	\begin{align*}
		f_t'(\mu_t s) &= \lambda \Var(V\mid \mu_t V + \sigma_t Z = \mu_t s)\\
		&\le \lambda \E[(V - v_i)^2 \mid \mu_t V + \sigma_t Z = \mu_t s] \\
		&\le 4\lambda B^2 \Pr[V \neq v_i \mid \mu_t V + \sigma_t Z = \mu_t s] \\
		&\le \frac{4\lambda B^2}{p_i} \exp\left(-\frac{3\mu_t \lambda\Delta^2}{10}\right)
	\end{align*}
	Let $g(\lambda) = \lambda^2 \exp(-C \lambda)$. 
	We see that $g'(\lambda) = (2\lambda - C\lambda^2) \exp(-C\lambda)$, so $g$ is maximized at $\lambda = \frac{2}{C}$, where it attains value $\left(\frac{2}{Ce}\right)^2$.
	Plugging back in the appropriate value of $C$, it follows that
	\[f_t'(\mu_t s)\le \frac{4B^2}{p_i \lambda}\cdot \left(\frac{20}{3e\mu_t \Delta^2}\right)^2 \le \frac 1\lambda \cdot \frac{28B^2}{\Delta^4 p_i}\]
	as desired (we also used that $\mu_t \ge \frac \lambda 2\ge 1$).
	
	It suffices to prove the probability statement.
	By Bayes' Theorem, we can express
	\[\frac{\Pr[V = v_j \mid \mu_t V + \sigma_t Z = \mu_t s]}{\Pr[V = v_i \mid \mu_t V + \sigma_t Z = \mu_t s]} = \frac{p_j\Pr[\mu_t v_j + \sigma_t Z = \mu_t s]}{p_i\Pr[\mu_t v_i + \sigma_t Z = \mu_t s]} = \frac{p_j}{p_i} \exp\left(-\frac{\mu_t^2}{2\sigma_t^2} \left((s - v_j)^2 - (s - v_i)^2\right)\right).\]
	Since $|s - v_i| \le \frac 15$, it follows that $(s - v_j)^2 - (s - v_i)^2 \ge (\Delta - \frac 15)^2 - (\frac 15)^2 = \Delta(\Delta - 0.4) \ge \frac{3\Delta^2}{5}$. Furthermore, we may simplify $\frac{\mu_t^2}{\sigma_t^2} = \mu_t \lambda$. On net, it follows that
	\[\Pr[V\neq v_i \mid \mu_t V + \sigma_t Z = \mu_t S] = \sum_{j\neq i} \Pr[V = v_j \mid \mu_t V + \sigma_t Z = \mu_t s] \le \frac{1 - p_i}{p_i}\exp\left(-\frac{3\mu_t\lambda\Delta^2}{10}\right)\]
	as desired.\\
	
	Finally, we prove nontriviality of iterates.
	Towards this, recall that $\frac 1n \|x^T\|_2^2 \rightarrow \E[f_{T-1}(\mu_{T-1} V + \sigma_{T-1} Z)^2]$ by State Evolution.
	By Jensen's inequality, it then follows that
	\[\E[f_{T-1}(\mu_{T-1} V + \sigma_{T-1} Z)^2] \ge \E[f_{T-1}(\mu_{T - 1} V + \sigma_{T-1} Z)]^2 = \E[V]^2\]
	where we used the optimality of $f_{T-1}$.
\end{proof}

\restateprop{prop:amp-classification}

\begin{proof}
	First, we show that the Nonnegative PCA iteration is $(1,1/\sqrt{2},\E[V] \sqrt{2/\pi})$-expanding (the first of these immediately follows from the Lipschitz constant of $\ReLu$). 
	We recall that $\frac 1n \|x^t\|_2^2 \rightarrow \sigma_t^2$.
	Further, recall that by convention, $\mu_0 = \lambda\E[V]$ and $\sigma_0 = 1$.
	
	First, note that
	\[\sigma_{t+1}^2 = \E[\ReLu(\mu_t V + \sigma_t Z)^2]\ge \frac 12 \E[\ReLu(\mu_t V + \sigma_t Z)^2 \mid Z\ge 0] \ge \frac 12 (\mu_t^2 + \sigma_t^2)\]
	which gives our value of $c_1$.
	Now, define $\rho_t = \frac{\mu_t}{\sigma_t}$ and note that $c_2^2 = \min \frac{\sigma_{t+1}^2}{\sigma_{t}^2} \ge \frac 12 \min_t \rho_t^2$ by the above. Note that
	\[\rho_{t+1}^2 = \frac{\mu_{t+1}^2}{\sigma_{t+1}^2} = \lambda^2 \frac{\E[V\ReLu(\mu_t V + \sigma_t Z)]^2}{\E[\ReLu(\mu_t V + \sigma_t Z)^2]} = \lambda^2 \frac{\E[V\ReLu(\rho_t V + Z)]^2}{\E[\ReLu(\rho_t V + Z)^2]}.\]
	We prove that this latter quantity is increasing in $\rho_t$. 
	Indeed, letting the numerator be $a^2(\rho)$ and the denominator $b(\rho)$, this is equivalent to showing that $2a'(\rho) b(\rho) \ge a(\rho) b'(\rho)$.
	We see that
	\[a'(\rho) = \E_V [V \cdot \E_Z[\ReLu'(\rho V + Z)]] = \E [V^2 \ind[\rho V + Z \ge 0]]\]
	and
	\[b'(\rho) = 2\E[V\ReLu(\rho V + Z)\cdot \ind[\rho V + Z\ge 0]].\]
	Thus, $2a(\rho) = b'(\rho)$ and so by Cauchy Schwarz we may now write
	\begin{align*}
		a(\rho) b'(\rho) &= 2\E[V\ind[\rho V + Z\ge 0] \ReLu(\rho V + Z)]^2\\
		&\le 2 \E[V^2 \ind[\rho V + Z\ge 0]] \cdot \E[\ReLu(\rho V + Z)^2]\\
		&= 2a'(\rho)\cdot b(\rho)
	\end{align*}
	as desired.
	Now, consider taking $\rho = 0$. Then, it follows that
	\[\rho_{t+1}^2 \ge \lambda^2 \frac{\E[V \ReLu(Z)]^2}{\E[\ReLu(Z)]^2} = \lambda^2 \E[V]^2 \cdot \frac{2}{\pi}\]
	which completes the proof of nnPCA being expanding.\\
	
	Next, let us recall that $\tanh(\lambda x)$ is the Bayes-AMP denoiser for $\Z_2$ synchronization, after the first iteration. 
	Then, applying~\pref{prop:bayes-controlled} and taking $p_i = \frac 12$, $B = 1$, and $\Delta = 2$ we see that we are $(1,\frac 72,\infty)$-controlled.
	
	To improve this latter guarantee, let us examine the structure of the denoiser $\tanh(\lambda x)$. 
	Recall that $\frac 1n \|x^{t+1}\|_2^2 \rightarrow \E[f_{t}(\mu_{t} V + \sigma_{t} Z)^2]$.
	
	We claim that for any $x\ge 0$, 
	\[\E[\tanh(x + Z)^2]\ge \E[\tanh(Z)^2]\]
	whenever $Z\sim N(0, \sigma^2)$ for some $\sigma > 0$. If so, it follows by symmetry that
	\[\E[f_{t}(\mu_{t} V + \sigma_{t} Z)^2] \ge \E[f_{t}(\sigma_{t} Z)^2] = \E[\tanh(\lambda \sigma_{t} Z)^2]. \]
	To prove this, let $g(x) = \E[\tanh(x + Z)^2]$ and note that $(\tanh(u)^2)'$ is an odd function.
	Then, we see that
	\[g'(x) = \int_{-\infty}^{\infty} \frac{\mathrm d}{\mathrm dx}\tanh(x + z)^2 \cdot \varphi_\sigma(z)\,\mathrm dz = \int_{-\infty}^{\infty} (\tanh(u)^2)' \varphi_\sigma (u -x)\,\mathrm du = \int_0^{\infty} (\tanh(u)^2)' \left(\varphi_\sigma (u - x) - \varphi_\sigma (u + x)\right)\,\mathrm du\]
	where $\varphi_\sigma$ is the PDF of a $N(0,\sigma^2)$.
	For positive $u$ and $x$, this quantity is non-decreasing (looking pointwise at the Gaussian PDF), which implies that $g(x) \ge g(0)$.

	Now, certainly, since $\tanh(\lambda \sigma_t Z)^2 \ge \tanh(\sigma_t Z)^2$, we may take $\lambda = 1$ to derive a bound, and we may assume that $\sigma_t \le 1$.
	
	Next, we see that $s^2 \tanh(x)^2 \le \tanh(s x)^2$ pointwise when $0 \le s \le 1$, since $\tanh$ is concave.
	Thus, it follows that
	\[\sigma_{t+1}^2 = \E[f_{t}(\mu_{t} V + \sigma_{t} Z)^2] \ge \sigma_t^2 \E[\tanh(Z)^2]\ge \frac{\sigma_t^2}{5}\]
	by analytically computing the latter expectation. 
	Since $\sigma_0 = 1$, it follows by induction that $\sigma_T^2 \ge \frac 1{5^T}$, yielding $(1,4,5^T)$-controlledness.\\

	For Bayes-optimal Sparse PCA, we have $B = \Delta = \frac{1}{\sqrt{\delta}}$, $\max \frac{1}{p_i} = B^2$, and $\E[V] = \sqrt \delta$. Thus, the iteration is $(\frac{1}{\sqrt{\delta}}, 28, \frac 1\delta)$-controlled.
\end{proof}

\subsection{Delocalization of AMP Iterations}\label{sec:amp-deloc}

We begin by proving~\pref{prop:deloc-amp}.

\restateprop{prop:deloc-amp}

\begin{proof}
	Let $\mu_t$ and $\sigma_t$ denote the signal strength and noise components, respectively. 
	We begin by proving the $\widetilde x^t$ statement.
	Recall by state evolution that
	\[\plim_{n\rightarrow\infty} \frac 1n \|\widetilde x^{t}\|_2^2 = \E[(\mu_t V + \sigma_t Z)^2] = \mu_t^2 + \sigma_t^2.\]
	Also note that for any $\theta > 0$, the test function $\psi(x) = x^2 \ind[x^2 \ge \theta]$ can be approximated arbitrarily well by $\PL(2)$ functions, which implies that state evolution holds for it (in particular, we can add a small linear component on $[\theta - \delta, \theta]$ for vanishing $\delta$).
	Let $F_i$ be the event that $|\widetilde x^{t}_i|$ is in the top $\eps n$ coordinates of $\widetilde x^{t}$ (breaking ties arbitrarily), and define $W = \mu_t V + \sigma_t Z$. 
	Then, we may write
	\begin{align*}
	 \frac 1n\max_{|S|=\eps n} \sum_{i\in S} (\widetilde x^{t}_i)^2 &= \frac 1n\sum_{i=1}^{n} (\widetilde x^{t}_i)^2 \ind[F_i]\\
	 &= \frac 1n\sum_{i=1}^{n} (\widetilde x^{t}_i)^2\ind[F_i \land (\widetilde x^{t}_i)^2 < \theta] + \frac 1n\sum_{i=1}^{n} (\widetilde x^{t}_i)^2\ind[F_i \land (\widetilde x^{t}_i)^2 \ge \theta]\\
	 &\le \frac 1n\sum_{i=1}^{n} \theta \ind[F_i] + \frac 1n\sum_{i=1}^{n} (\widetilde x^{t}_i)^2 \ind[(\widetilde x^{t}_i)^2 \ge \theta]\\
	 &\rightarrow  \eps \theta + \E[W^2 \ind[W^2 \ge \theta]].
	\end{align*}
	Next, note that $W - \mu_t \E[V]$ is centered and subgaussian with parameter $\sigma \sqrt{\mu_t^2 + \sigma_t^2}$. 
	To bound the expectation, we may directly evaluate
	\[\E[W^2 \ind[W^2 \ge \theta]] = \int_0^{\infty}\Pr[W^2 \ind[W^2 \ge \theta] \ge t]\,\mathrm dt = \int_0^{\theta} \Pr[W^2 \ge \theta]\,\mathrm dt + \int_{\theta}^{\infty} \Pr[W^2 \ge t]\,\mathrm dt.\]
	We will choose $\theta = 8\sigma^2 (\mu_t^2 + \sigma_t^2) \cdot \log \frac 1\eps$, and notice that $\sqrt \theta - \mu_t \E[V]\ge \frac{\sqrt\theta}{2}$ (note that $|\E[V]| \le 1$ by Cauchy Schwarz). This implies that for any $t\ge \theta$, by the subgaussianity of $W$, we know that
	\[\Pr[W^2 \ge t] \le 2\exp\left(-\frac{(\sqrt t - \mu_t \E[V])^2}{2\sigma^2 (\mu_t^2 + \sigma_t^2)}\right) \le 2\exp\left(-\frac{t}{8\sigma^2(\mu_t^2 + \sigma_t^2)}\right).\]
	For simplicity, we let $\gamma = 8\sigma^2 (\mu_t^2 + \sigma_t^2)$, and so $\theta = \gamma \log \frac 1\eps$.
	Plugging this back in to the integral, it follows that
	\[\E[W^2 \ind[W^2 \ge \theta]] \le 2\theta \exp\left(-\frac{\theta}{\gamma}\right) + 2\gamma\cdot \left(\left.-\exp\left(-\frac{t}{\gamma}\right)\right|_{t=\theta}^{\infty}\right) = 2(\gamma\log \frac 1\eps + \gamma)\exp\left(-\frac{\gamma \log \frac 1\eps}{\gamma}\right) \le 3\gamma \eps \log \frac 1\eps.\]
	Therefore, on aggregate, it follows that
	\[\frac 1n \max_{|S|=\eps n} \sum_{i\in S} (\widetilde x^{t}_i)^2 \le \gamma \eps \log \frac 1\eps + 3\gamma\eps \log \frac 1\eps = 32\sigma^2 (\mu_t^2 + \sigma_t^2)\eps \log \frac 1\eps = 32\sigma^2 \eps \log \frac 1\eps \cdot \left(\frac 1n \|\widetilde x^t\|_2^2\right).\]
	Next, note that for any set $S$, $\sum_{i\in S} (x^{t+1}_i)^2 \le L^2 \sum_{i\in S} (\widetilde x^t_i)^2$, and thus
	\[\frac 1n \max_{|S|=\eps n} \sum_{i\in S} (x^{t+1}_i)^2 \le L^2 \cdot \frac 1n \max_{|S|=\eps n} \sum_{i\in S} (\widetilde x^{t}_i)^2 = 32L^2 \sigma^2 \eps \log \frac 1\eps \cdot \left(\frac 1n \|\widetilde x^t\|_2^2\right).\]
	We finish by noting that $\|x^{t+1}\|_2^2 \ge c_1^2 \|\widetilde x^t\|_2^2$.
\end{proof}

Next, we prove that large enough principal minors of $X$ continue to enjoy properties of the top eigenvector such as delocalization. Recall that $\alpha \triangleq C\sigma^2 \eps \log \frac 1\eps$.

\begin{proposition}[Top eigenvector of principal minor is delocalized and close to $\varphi_1(X)$]\torestate{
\label{prop:deloc-top-eig}
	Suppose that $X$ is spiked Gaussian, and the signal distribution $V$ is $\sigma$-subgaussian. Then, with high probability over $X$, the following statements are true for \emph{all} principal minors $A$ of $X$ with at least $(1 - c\eps)n$ rows:
	\begin{itemize}
		\item $A$ has large quadratic form with $\varphi_1(X)$: $\varphi_1(X)^\top A \varphi_1(X) \ge \widetilde \lambda (1 - 3c\alpha)$.
		\item $A$ has large top eigenvalue: $\|A\|_{\op} \ge \widetilde \lambda(1 - 3c\alpha)$.
		\item The top eigenvector of $A$ is delocalized: $\deloc_{\eps}(\varphi_1(A)) \le 27\max(c, 1)\alpha \frac{\widetilde \lambda}{\widetilde \lambda - 2}$.
	\end{itemize}
}
\end{proposition}

\begin{proof}
Let $x = \varphi_1(X)$ for ease of notation, and suppose that $T$ is the collection of rows contained in $A$.
	By~\cite[Proposition 3.4]{FVRS22}, we know that $x$ satisfies state evolution, with $\mu_0 = \sqrt{1 - \lambda^{-2}}$ and $\sigma_0 = \frac 1\lambda$ (or scaled versions of these, but delocalization is scale-invariant). 
	Thus, by~\pref{prop:deloc-amp}, it follows that $x$ is $(\alpha, \eps)$-delocalized.
	
	From here, recall that $Xx = \widetilde \lambda x$, which implies that $X_{T,T} x_T + X_{T, \overline T} x_{\overline T} = \widetilde \lambda x_T$ and thus $x_T^\top X x_T + x_T^\top X x_{\overline T} = \widetilde \lambda \|x_T\|_2^2$.
	
	Therefore, we may write 
	\[x_T^\top X x_T = x^\top X x - 2x_T^\top X x_{\overline T} - x_{\overline T}^{\top} X x_{\overline T} = x^\top X x - 2(\widetilde \lambda \|x_T\|_2^2  - x_T^\top X x_T) - x_{\overline T}^\top X x_{\overline T}.\]
	
	Rearranging to move $w_T^\top X w_T$ to the same side, we have that
	\[x_T^\top X x_T = 2\widetilde \lambda \|x_T\|_2^2 - x^\top X x + x_{\overline T}^{\top} X x_{\overline T} \ge \widetilde \lambda (2\|x_T\|_2^2 - 1 - c\alpha) \ge \widetilde \lambda (2(1 - c\alpha) - 1 - c\alpha) = \widetilde \lambda (1 - 3c\alpha).\]
	Since $x^\top A x = x_T^\top X x_T$, this completes the first claim.
	The second claim is trivial given the first, since $\|A\|_{\op} \ge x^\top A x$.

	For the third claim, let $\delta = 3c\alpha \frac{\widetilde \lambda}{\widetilde \lambda - 2}$ and $y = \varphi_1(A)$, so $y^\top X y \ge \widetilde \lambda (1 - 3c\alpha)$.
	Expanding $y$ in the $\varphi_i(X)$ basis as $y = \sum \beta_i \varphi_i(X)$, it follows that
	\[\widetilde \lambda (1 - 3c\alpha) \le y^\top X y = \sum \beta_i^2 \lambda_i(X) \le \beta_1^2 \widetilde \lambda + (1 - \beta_1^2)\cdot 2 = 2 + \beta_1^2 (\widetilde \lambda - 2).\]
	Rearranging, it follows that $\langle x, y\rangle^2 = \beta_1^2\ge 1 - 3c\alpha \frac{\widetilde \lambda}{\widetilde \lambda - 2} = 1 - \delta$. 
	
	Now, suppose for contradiction that $y$ was not $(9\delta,\eps)$-delocalized: then, there must exist some subset $S$ such that $\|y_S\|_2^2 > 9\delta$.
	For this subset, we can write
	\begin{align*}
	\langle x,y\rangle^2 &= \left(\langle x_S, y_S\rangle + \langle x_{\overline S}, y_{\overline S}\rangle \right)^2 \\
	&\le \left(\|x_S\|_2 \|y_S\|_2 + \sqrt{1 - \|x_S\|_2^2}\sqrt{1 - \|y_S\|_2^2}\right)^2
	\end{align*}
	Note that the equality condition of Cauchy Schwarz is achieved when $\|x_S\|_2 = \|y_S\|_2$, and therefore this value is maximized when $\|x_S\|_2$ and $\|y_S\|_2$ are as close as possible. 
	We can bound
	\begin{align*}
	\langle x,y\rangle^2 &< \left(\sqrt{9\alpha\delta} + \sqrt{(1 - \alpha)(1 - 9\delta)}\right)^2\\
	&=9\alpha \delta + (1 - \alpha)(1 - 9\delta) + 2\sqrt{9\alpha \delta(1 - \alpha)(1 - 9\delta)}\\
	& \le 1 -\alpha - 3\delta + 18\alpha \delta + 2\sqrt{9\alpha \delta}\\
	& = 1 + 6\alpha\delta - (3\sqrt \delta - \sqrt{\alpha})^2\\
	&\le 1 + 6\alpha \delta - 4\delta\\
	&\le 1 - 3\delta
	\end{align*}
	whenever $\alpha \le \frac 16$.
	Chaining these two inequalities then gives that $\langle x, y\rangle^2 < 1 - 3\delta < 1 - \delta \le \langle x,y\rangle^2$, which is the desired contradiction.
\end{proof}

This enables us to extend the above proposition to corruptions of $X$.

\restateprop{prop:deloc-eig}

\begin{proof}[Proof]
Let $x = \varphi_1(X)$ and write $\hat{Y} = X_{T,T} + E_{T,T}$.
Note that $\|E_{T,T}\|_{\op} \le \|\hat Y\|_{\op} + \|X_{T,T}\|_{\op} \le 6\widetilde \lambda + o_n(1)$.
By~\pref{prop:deloc-top-eig} and the fact that $\deloc_{\eps}(x) \le \alpha$, we see that
\[x^\top \hat{Y} x = x^\top X_{T,T} x + x^\top E_{T,T} x \ge \widetilde \lambda (1 - 3c\alpha) - \|E_{T,T}\|_{\op} \|x_{S^\ast\cap T}\|_2^2 \ge \widetilde \lambda (1 - 3c\alpha - 6\alpha)\]
as desired.	
\end{proof}

We also show the following toy bound on the delocalization of a $\sigma$-subgaussian vector.

\begin{proposition}
\label{prop:toy-deloc}
	Suppose that $V$ is a $\sigma$-subgaussian, mean $0$, random variable. Then, If $v$ is such that $v_i \sim V$ i.i.d, it follows that $v$ is $(C\sigma^2 \eps \log \frac 1\eps,\eps)$-delocalized with high probability.
\end{proposition}

\begin{proof}
	We directly use a union bound over all $\binom{n}{\eps n}$ possible subsets.
	Then, for any particular set $S$, we know by Bernstein's inequality (since $v_i^2$ is $2\sigma^2$-subexponential)
	\[\Pr\left[\sum_{i\in S} v_i^2 \ge C\sigma^2 \eps \log \frac 1\eps n\right ] \le \exp\left(-\frac{c(C\sigma^2 \eps \log \frac 1\eps n)}{2\sigma^2}\right) = \exp\left(-cCn\eps \log \frac 1\eps\right). \]
	Since $\binom{n}{\eps n} \le \exp(n\eps \log \frac e\eps)$, we can choose $C\ge \frac 2c$ to finish.
\end{proof}

With these done, we are now ready to prove~\pref{prop:deloc-off-corruption}.

\restateprop{prop:deloc-off-corruption}

\begin{proof}
Let $S$ be a collection of at most $\eps n$ (nonempty) rows and let $T$ be the remaining nonempty rows (note that $|T| \ge (1 - (c + 1)\eps)n$). We will assume that $|S| > 0$, since in fact~\pref{prop:deloc-top-eig} proves the claim when $|S| = 0$.

Let $\mu$ be the top eigenvalue of $B$.
By~\pref{prop:deloc-top-eig}, we know that $\mu \ge \lambda_1(B_{T,T}) \ge \widetilde \lambda (1 - (6 + 6c)\alpha)$.
We may write $Bw = \mu w$, which we can expand to $B_{T,T} w_T + B_{T,S} w_S = \mu w_T$. 
Next, write $u = \varphi_1(B_{T,T})$ and $w_T = \beta u + u^\perp$. 
Let $\Pi_{u^\perp} = I - uu^\top$. Then, we may write
\begin{align*}
	\Pi_{u^\perp}(B_{T,T} w_T + B_{T,S} w_S) &= \Pi_{u^\perp}\mu w_T & \implies \\
	\Pi_{u^\perp} B_{T,T} \Pi_{u^\perp} u^\perp + \Pi_{u^\perp} B_{T,S} w_S &= \mu \Pi_{u^\perp} w_T &\implies\\
	u^\perp &= \Pi_{u^\perp} (\mu I - B_{T,T})^{-1} \Pi_{u^\perp} B_{T,S} w_S.
\end{align*}
Let $P = \Pi_{u^\perp} (\mu I - B_{T,T})^{-1} \Pi_{u^\perp}$ for simplicity.
By eigenvalue interlacing, we know that $\lambda_2(B_{T,T}) \le \lambda_2(X) = 2 + o_n(1)$.
Thus, it follows that $\|P\|_{\op} \le \frac{1}{\mu - 2 - o_n(1)}$ for all $T$ simultaneously.

Next, we expand $B_{T,S} = \frac \lambda n v_T v_S^\top + G_{T,S}$, and write $u_1 = \frac \lambda n P v_T v_S^\top w_s$ and $u_2 = P G_{T,S} w_s$.
We see that 
\[\deloc_{\eps}(w_T) \le 3(\deloc_{\eps}(u) + \deloc_{\eps}(u_1) + \deloc_{\eps}(u_2))\]
by the Almost-Triangle Inequality, so it suffices to show that each of these terms is properly bounded.
The first term follows from~\pref{prop:deloc-top-eig} applied to $X_{T,T}$, and thus $\deloc_{\eps}(u) \le 27 (c + 1) \alpha \frac{\widetilde \lambda}{\widetilde \lambda - 2}$.

For the second term, we can naively bound $\deloc_{\eps}(u_1) \le \deloc_1(u_1)$ and apply~\pref{prop:toy-deloc} to see that
\begin{align*}
	\deloc_{\eps}(u_1) &\le \frac{\lambda^2}{n^2} \|P v_T v_S^\top w_S\|_2^2\\
	& = \frac{\lambda^2}{n^2} \langle v_S, w_S\rangle^2 \|P v_T\|_2^2\\
	& \le \frac{\lambda^2}{n^2 }\deloc_{\eps}(v) \cdot \|w_S\|_2^2 \cdot \|P\|_{\op}^2 \cdot \|v_T\|_2^2\\
	&\le C\alpha \left(\frac{\lambda}{\mu - 2}\right)^2.
\end{align*}

What remains is to handle $\deloc_{\eps}(u_2)$. We do so by union bounding over all possible $S$ and $T$, and additionally discretizing $\mu$ up to error $\frac{\widetilde \lambda}{n^2}$ (note that $\mu \le \|E\|_{\op} + \|X\|_{\op} \le 7\widetilde \lambda$, so the number of such $\mu$ is at most $O(n^2)$).
Doing this discretization gives $\|\hat{P}\|_{\op} = (1\pm o_n(1)) \|P\|_{\op}$.

Given a fixed $S, T, \mu$, it follows that $G_{T,S}$ is independent of $P$.  
Furthermore, note that we can, without loss of generality, assume that $\|w_S\|_2^2 = 1$. 
Now, take a $1/4$-net $\calN$ over possible $w_S$: this net has size $9^{|S|} \le \exp(c|S|)$.
We are interested in $\max_{|U|\le \eps n} \|\Pi_U P G_{T,S} w_S\|_2^2$ (where $\Pi_U$ is the projector to the coordinates in $U$), so fix a particular such set and let $\calM$ be a $1/4$-net over $z$ supported only on $U$, thus of size $9^{\eps n}$. Then, we see that
\[\|\Pi_U P G_{T,S} w_S \|_2^2 \le \|\Pi_U P G_{T,S}\|_{\op}^2 \le 4\max_{z_1\in \calN, z_2\in \calM} (z_1^\top PG_{T,S} z_2)^2.\]

Recalling that $z_1^\top PG_{T,S} z_2\sim N(0, \frac 1n \|P z_1\|_2^2 \|z_2\|_2^2)$ and noting that this variance term is at most $\frac{1}{(\mu - 2)^2}$, it follows that
\[\Pr[(z_1^\top PG_{T,S} z_2)^2 > t] \le 2\exp\left(-\frac{n(\mu - 2)^2 t}{2}\right).\]
The total size of our union bound is at most (using script symbols to denote collections of their respective sets)
\[|\calR| \cdot |\calS|\cdot |\calU|\cdot |\calN|\cdot |\calM| \cdot O(n^2) \le \exp\left(3c\eps n \log \frac 1\eps  + \eps n + C\log n\right) \le \exp\left(5c\eps n \log \frac 1\eps\right).\]

Buoyed by this, we choose $t = 20c \eps \log \frac 1\eps \cdot \frac{1}{(\mu - 2)^2} \le 20c\alpha \frac{1}{(\mu - 2)^2} $, which implies that $\deloc_{\eps}(u_2) \le 4t$ with high probability.

Therefore, it follows that
\[\deloc_{\eps}(w_T) \le 3\left(27(c + 1) \alpha \frac{\widetilde \lambda}{\widetilde \lambda - 2} + C\alpha \left(\frac{\lambda}{\mu - 2}\right)^2 + 20c\alpha \frac{1}{(\mu - 2)^2}\right) \le Cc\alpha \left(\frac{\lambda}{\mu - 2}\right)^2\]
for an appropriate constant $C$ and noting that $\lambda > 1$ (thus $\widetilde \lambda \le 2\lambda$).
Now, recall that $\mu \ge \widetilde \lambda (1 - (6 + 6c)\alpha)$. Thus, we can write
\[\mu - 2\ge \widetilde \lambda (1 - (6 + 6c)\alpha) - 2 = (\widetilde \lambda - 2)(1 - (6 + 6c)\alpha) - (12 + 12c)\alpha.\]
Recalling that $\frac{1}{x - \delta} \le \frac{2}{x}$ whenever $x \ge 2\delta$ and $\delta \le \frac 12$ (since $2(x-\delta) \ge x$), it follows that for small enough $\eps \ll \widetilde \lambda - 2$ (and thus, $\alpha$),
\[\frac{\lambda}{\mu - 2} \le \frac{\lambda}{\widetilde \lambda - 2}\cdot \frac{2}{1 - (6 + 6c)\alpha} \le 4\frac{\widetilde \lambda}{\widetilde \lambda - 2}.\]
This completes the proof as it implies that $\deloc_{\eps}(w_T) \le Cc\alpha \cdot \left(\frac{\widetilde \lambda}{\widetilde \lambda - 2}\right)^2$.
\end{proof}

\end{document}